\documentclass[draftclsnofoot,onecolumn]{IEEEtran}

\usepackage{amssymb,amsfonts,amsmath,amstext,mathrsfs}
\usepackage{latexsym}
\usepackage{epsfig,psfrag,color,graphics,epsf}
\usepackage{enumerate,cite}

\usepackage{amsmath}
\usepackage{amssymb}
\usepackage{amsfonts}
\usepackage{amsthm}
\usepackage{graphicx}
\usepackage{float}
\usepackage{graphics}
\usepackage{graphicx}
\usepackage{euscript}
\usepackage{psfrag}

\newtheorem{theorem}{Theorem}
\newtheorem{lemma}{Lemma}

\newtheorem{corollary}{Corollary}

\newtheorem{proposition}{Proposition}

\newcommand {\bv} {\mbox{\boldmath $v$}}

\newcommand {\bx} {\mbox{\boldmath $x$}}
\newcommand {\by} {\mbox{\boldmath $y$}}

\newcommand {\bE} {\mbox{\boldmath $E$}}

\newcommand {\bX} {\mbox{\boldmath $X$}}
\newcommand {\bY} {\mbox{\boldmath $Y$}}

\newcommand{\calC}{{\cal C}}
\newcommand{\calD}{{\cal D}}
\newcommand{\calE}{{\cal E}}

\newcommand{\calG}{{\cal G}}

\newcommand{\calK}{{\cal K}}
\newcommand{\calL}{{\cal L}}

\newcommand{\calP}{{\cal P}}

\newcommand{\calR}{{\cal R}}

\newcommand{\calX}{{\cal X}}
\newcommand{\calY}{{\cal Y}}

\begin{document}

\sloppy

\title{On Achievable Rates for Channels with Mismatched Decoding 
}

\author{
  Anelia Somekh-Baruch\thanks{A.\ Somekh-Baruch is with the Faculty of Engineering at Bar-Ilan University, Ramat-Gan, Israel.  Email: somekha@biu.ac.il. A shorter version of this paper was accepted to the International Symposium on Information Theory 2013. This paper was submitted to the IEEE Transactions on Information Theory.}
}
\maketitle

\begin{abstract}

The problem of mismatched decoding for discrete memoryless channels is addressed. A mismatched cognitive multiple-access channel is introduced, and an inner bound on its capacity region is derived using two alternative encoding methods: superposition coding and random binning. 
The inner bounds are derived by analyzing the average error probability of the code ensemble for both methods and by a tight characterization of the resulting error exponents. Random coding converse theorems are also derived. 
A comparison of the achievable regions shows that in the matched case, random binning performs as well as superposition coding, i.e., the region achievable by random binning is equal to the capacity region. 
The achievability results are further specialized to obtain a lower bound on the mismatch capacity of the single-user channel by investigating a cognitive multiple access channel whose achievable sum-rate serves as a lower bound on the single-user channel's capacity. 
In certain cases, for given auxiliary random variables this bound strictly improves on the achievable rate derived by Lapidoth.
\end{abstract}

\vspace{9cm}

\pagebreak

\section{Introduction}\label{sc: Introduction}

The mismatch capacity is the highest achievable rate using a given decoding rule.  Ideally, the decoder uses the maximum-likelihood rule which minimizes the average probability of error, or other asymptotically optimal decoders such as the joint typicality decoder, or the Maximum Mutual Information (MMI) decoder \cite{CsiszarKorner81}. The mismatch capacity reflects a practical situation in which due to inaccurate knowledge of the channel, or other practical limitations, the receiver is constrained to use a possibly suboptimal decoder. This paper focuses on mismatched decoders that are defined by a mapping $q$, which for convenience will be referred to as "metric" from the product of the channel input and output alphabets to the reals. The decoding rule maximizes, among all the codewords, the accumulated sum of metrics between the channel output sequence and the codeword.

Mismatched decoding has been studied extensively for discrete memoryless channels (DMCs).
A random coding lower bound on the mismatched capacity was derived by Csisz\'{a}r and K{\"o}rner
and by Hui \cite{CsiszarKorner81graph}, \cite{Hui83}. 
Csisz\'{a}r and Narayan \cite{CsiszarNarayan95} showed that the random coding bound is not tight.
They established this result by proving that the random coding bound for the product channel $P_{Y_1,Y_2|X_1,X_2}=P_{Y_1|X_1}\times P_{Y_2|X_2}$ (two consecutive channel uses of channel $P_{Y|X}$) may result in higher achievable rates. 
Nevertheless, it was shown in \cite{CsiszarNarayan95} that the positivity of the random coding lower-bound is a necessary condition for a positive mismatched capacity. A converse theorem for the mismatched binary-input DMC was proved in \cite{Balakirsky95}, but in general, the problem of determining the mismatch capacity of the DMC remains open.

Lapidoth \cite{Lapidoth96} introduced the mismatched multiple access channel (MAC) and derived an inner bound on its capacity region. The study of the MAC case led to an improved lower bound on the mismatch capacity of the single-user DMC by considering the maximal sum-rate of an appropriately chosen mismatched MAC whose 
codebook is obtained by expurgating codewords from the product of codebooks of the two users. 

In 
 \cite{CsiszarKorner81graph}, 
an error exponent for
random coding with fixed composition codes and mismatched decoding was established using a
graph decomposition theorem. In a recent work, Scarlett and {Guill\'{e}n i F\`{a}bregas} \cite{ScarlettFabregas2012} characterized the achievable error exponents obtained by a constant-composition random
coding scheme for the MAC. 
For other related works and extensions see 
\cite{Balakirsky_conference_95,MerhavKaplanLapidothShamai94,LiuHughes96,Lapidoth96b,GantiLapidothTelatar2000,ScarlettAlfonsoFabregas2013,ScarlettMartinezGuilleniFabregas2012AllertonSU} and references therein.

This paper introduces the {\it cognitive} mismatched two-user multiple access channel. 
The matched counterpart of this channel is in fact a special case of the MAC with a common message studied by Slepian and Wolf \cite{SlepianWolf73}. Encoder $1$ shares the message index   it wishes to transmit with Encoder $2$ (the cognitive encoder), and the latter transmits an additional message to the same receiver. 
 Two achievability schemes for this channel with a mismatched decoder are presented. The first scheme is based on superposition coding and the second uses random binning.
The achievable regions are compared, and an example is shown in which the achievable region obtained by random binning is strictly larger than the rate-region achieved by superposition coding. 
In general it seems that neither achievable region dominates the other, and conditions are shown under which random binning is guaranteed to perform at least as well as supposition coding in terms of achievable rates and vice versa. As a special case it is shown that in the matched case, where it is well known that superposition coding is capacity-achieving, binning also achieves the capacity region. The resulting region of the cognitive mismatched MAC achievable by binning in fact contains the mismatched non-cognitive MAC achievable region studied by Lapidoth \cite{Lapidoth96}. Although this is not surprising, in certain cases, for fixed auxiliary random variables cardinalities, it serves to derive an improved achievable rate for the mismatched single-user channel.

The outline of this paper is as follows. Section \ref{sc: Notation} presents notation conventions. Section \ref{sc:  Preliminaries} provides
some necessary background in more detail. 
Section \ref{sc: The Mismatched Cognitive MAC} introduces the mismatched cognitive MAC and presents the achievable regions. 
Section \ref{sc: discussion} is devoted to discussing the results pertaining to the mismatched cognitive MAC. 
The following section 
\ref{sc: The Implication for the Single-User Channel} presents a lower bound on the capacity of the single-user mismatched DMC. 
Section \ref{eq: Conclusion} develops the concluding remarks.
Finally, the proofs of the main results appear in Appendices \ref{sc: Outline of the proof of Theorem 2}-\ref{sc: first equivalence}.

\section{Notation and Definitions}\label{sc: Notation}

Throughout this paper, scalar random variables are denoted by capital letters, their 
sample values are denoted by the respective lower case letters, and their alphabets are denoted by their
 respective calligraphic letters, e.g.\ $X$, $x$, and $\calX$, respectively. A similar convention applies to random vectors of 
dimension $n$ and their sample values, which are denoted with the same symbols in the boldface font, e.g., $\bx=(x_1,...x_n)$. 
The set of all $n$-vectors with components taking values in a certain finite alphabet are denoted 
by the same alphabet superscripted by $n$, e.g., $\calX^n$.

Information theoretic 
quantities such as entropy, conditional entropy, and mutual information are denoted following the usual conventions in 
the information theory literature, e.g., $H (X )$, $H (X |Y )$, $I(X;Y)$ and so on. To emphasize 
the dependence of the quantity on a certain underlying probability distribution, say $\mu$, it is subscripted 
by $\mu$, i.e., with notations such as $H_\mu(X )$, $H_\mu(X |Y)$, $I_\mu(X;Y)$, etc. 
The divergence (or Kullback -Liebler distance) between two probability measures $\mu$ and $p$ is denoted by $D(\mu\| p)$, and when there is a need to make a distinction between $P$ and $Q$ as joint distributions of $(X,Y)$ as opposed to the corresponding marginal distributions of, say, $X$, subscripts are used to avoid ambiguity, that is, the notations $D(Q_{XY}\|P_{XY})$ and $D(Q_X\|P_X)$.The expectation operator is denoted 
by $\bE \{\cdot\}$, and once again, to make the dependence on the underlying distribution 
$\mu$ clear, it is denoted by $\bE_\mu \{\cdot\}$. 
The cardinality of a finite set $A$ is denoted by $|A|$. 
The indicator function of an event $\calE$ is denoted by $1\{\calE \}$. 

Let $\calP(\calX)$ denote the set of all probability measures on $\calX$. 
For a given sequence $\by \in \calY^n$, $\calY$ being a finite alphabet,  $\hat{P}_{\by}$ denotes the empirical distribution on $\calY$ extracted from $\by$, in other words, $\hat{P}_{\by}$ is the vector $\{ \hat{P}_{\by} (y), y\in\calY\}$, where $ \hat{P}_{\by} (y)$ is the relative frequency of the letter $y$ in the vector $\by$. The type-class of $\bx$ is the set of $\bx'\in\calX^n$ such that $\hat{P}_{\bx'}=\hat{P}_{\bx}$, which is denoted $T(\hat{P}_{\bx})$.
The conditional type-class of $\by$ given $\bx$ is the set of $\tilde{\by}$'s such that $\hat{P}_{\bx,\tilde{\by}}=\hat{P}_{\bx,\by}=Q_{X,Y}$, which is denoted $T(Q_{X,Y} |\bx)$ with a little abuse of notation. The set of empirical measures of order $n$ on alphabet $\calX$ is denoted  $\calP_n(\calX)$.

For two sequences of positive numbers, $\{a_n\}$ and $\{b_n\}$, the notation $a_n\doteq b_n$ means that $\{a_n\}$ and $\{b_n\}$ are of the same exponential order, i.e., $\frac{1}{n} \ln \frac{a_n}{b_n}\rightarrow 0$ as $n\rightarrow \infty$. Similarly, $a_n\stackrel{\cdot}{\leq}b_n$ means that $\limsup_n \frac{1}{n} \ln \frac{a_n}{b_n} \leq 0$, and so on. Another notation is that for a real number $x$, $|x|^+=\max\{0,x\}$. 

Throughout this paper logarithms are taken to base $2$.

\section{Preliminaries} \label{sc:  Preliminaries}

Consider a DMC with a finite input alphabet 
$\calX$ and finite output alphabet $\calY$, which is governed by the conditional p.m.f.\ $P_{Y|X}$. As the channel is fed by an input vector $\bx \in\calX^n$, it generates an output vector $\by \in\calY^n$ according to the sequence of conditional probability distributions 
\begin{equation}P (y_i |x_1 , . . . , x_i , y_1 , . . . , y_{i-1} ) = P_{Y|X}(y_i|x_i), \quad i = 1, 2, . . . , n\end{equation}
where for $i = 1, (y_1 , . . . , y_{i-1}) $ is understood as the null string. A rate-$R$ block-code of length $n$ consists of $ 2^{nR}$  $n$-vectors $\bx(m)$, $m = 1, 2, . . . , 2^{nR}$, which represent $2^{nR}$ different messages, i.e., it is defined by the encoding function
\begin{flalign}
f_n:\; \{1,...,2^{nR}\} \rightarrow \calX^n.
\end{flalign}
It is assumed that all possible messages are a-priori equiprobable, i.e., $P (m) = 2^{-nR}$ for all 
$m$, and denote the random message by $W$.

A mismatched decoder for the channel is defined by a mapping 
\begin{flalign}
q_n:\;  \calX^n\times \calY^n\rightarrow  \mathbb{R},
\end{flalign}
where the decoder declares that message $i$ was transmitted iff 
\begin{flalign}
q_n(\bx(i),\by)>q_n(\bx(j),\by), \forall j\neq i,
\end{flalign}
and if no such $i$ exists, an error is declared. The results in this paper refer to the case of additive decoding functions, i.e.,  
\begin{flalign}
 q_n(x^n,y^n)=\frac{1}{n}\sum_{i=1}^n q(x_i,y_i), 
\end{flalign}
where $q$ is a mapping from $\calX\times \calY$ to $\mathbb{R}$.

A rate $R$ is said to be achievable for the channel $P_{Y|X}$ with a decoding metric $q$ if there exists a sequence of codebooks $\calC_n,\; n\geq 1$ of rate $R$   
such that the average probability of error incurred by the decoder $q_n$ applied to the codebook $\calC_n$ and the channel output vanishes as $n$ tends to infinity.
The capacity of the channel with decoding metric $q$ is the supremum of all achievable rates.

The notion of mismatched decoding can be extended to a MAC $P_{Y|X_1,X_2}$ with codebooks $\calC_{n,1}=\{\bx_1(i)\}, i=1,....,2^{nR_1}$, $\calC_{n,2}=\{\bx_2(j)\}, j=1,....,2^{nR_2}$.
A mismatched decoder for a MAC is defined by the mapping 
\begin{flalign}\label{eq: q_n MAC dfn}
q_n&:\; \calX_1^n\times\calX_2^n\times \calY^n \rightarrow  \mathbb{R},
\end{flalign}
where similar to the single-user's case, the decoder outputs the messages $(i,j)$ iff for all $(i',j')\neq (i,j)$
\begin{flalign}\label{eq: q_n MAC}
q_n(\bx_1(i),\bx_2(j),\by)> q_n(\bx_1(i'),\bx_2(j'),\by).
\end{flalign}
The focus here is on additive decoding functions, i.e.,  
\begin{flalign}\label{eq: qn additive MAC}
q_n(x_1^n,x_2^n,y^n)=\frac{1}{n}\sum_{i=1}^n q(x_{1,i},x_{2,i},y_i),
\end{flalign}
where $q$ is a mapping from $\calX_1\times\calX_2\times \calY$ to $\mathbb{R}$.
The achievable rate-region of the MAC $P_{Y|X_1,X_2}$ with decoding metric $q$ is the closure of the set of rate-pairs $(R_1,R_2)$ for which there exists a sequence of codebooks $\calC_{n,1},\calC_{n,2}$, $n\geq 1$ of rates $R_1$ and $R_2$, respectively, such that the average probability of error that is incurred by the decoder $q_n$ when applied to the codebooks $\calC_{n,1},\calC_{n,2}$ and the channel output vanishes as $n$ tends to infinity.

Before describing the results pertaining to the cognitive MAC, we state the best known inner bound on the capacity region of the mismatched (non-cognitive) MAC which was introduced in \cite{Lapidoth96}. The inner bound is given by $\calR_{LM}$ where
\begin{flalign}\label{eq: calRLM_dfn}
\calR_{LM}=&\mbox{closure of the CH of } \underset{P_{X_1},P_{X_2}}{\cup}  \bigg\{(R_1,R_2):\nonumber\\
R_1&<\tilde{R}_1 =\min_{f\in \calD_{(1)}} I_f(X_1;Y|X_2)+I_f(X_1;X_2)\nonumber\\
R_2&<\tilde{R}_2=\min_{f\in \calD_{(2)}} I_f(X_2;Y|X_1)+I_f(X_1;X_2)\nonumber\\
R_1+R_2&<\tilde{R}_0=\min_{f\in \calD_{(0)}} I_f(X_1,X_2;Y)+I_f(X_1;X_2)\bigg\},
\end{flalign}
where CH stands for "convex hull",
\begin{flalign}
\calD_{(1)}&=\{f_{X_1,X_2,Y}:\; f_{X_1}=P_{X_1}, \nonumber\\
&\quad \quad f_{X_2,Y}=P_{X_2,Y}, \bE_f(q)\geq \bE_P(q)\} \nonumber\\
\calD_{(2)}&=\{f_{X_1,X_2,Y}:\; f_{X_2}=P_{X_2},  \nonumber\\
&\quad \quad f_{X_1,Y}=P_{X_1,Y}, \bE_f(q)\geq \bE_P(q)\}
\nonumber\\
\calD_{(0)}&=\{f_{X_1,X_2,Y}:\; f_{X_1}=P_{X_1},\nonumber\\
&\quad \quad  f_{X_2}=P_{X_2},  f_{Y}=P_{Y}, \bE_f(q)\geq \bE_P(q),\nonumber\\
&\quad \quad  I_f(X_1;Y)\leq R_1, I_f(X_2;Y)\leq R_2 \}.\label{eq: calD dfn}
\end{flalign}
and where $P_{X_1,X_2,Y}=P_{X_1}\times P_{X_2}\times P_{Y|X_1,X_2}$.

\section{The Mismatched Cognitive MAC}\label{sc: The Mismatched Cognitive MAC}

The two-user discrete memoryless cognitive MAC is defined by the input alphabets $\calX_1$, $\calX_2$, output alphabet $\calY$ and conditional transition probability $P_{Y|X_1,X_2}$. 
A block-code of length $n$ for the channel is defined by the two encoding mappings
\begin{flalign}
f_{1,n}&:\; \{1,...,2^{nR_1}\}\rightarrow \calX_1^n\nonumber\\
f_{2,n}&:\; \{1,...,2^{nR_1}\}\times \{1,...,2^{nR_2}\} \rightarrow \calX_2^n,
\end{flalign}
resulting in two codebooks $\{\bx_1(i)\}, i=1,....,2^{nR_1}$ and $\{\bx_2(i,j)\}, i=1,....,2^{nR_1},  j=1,....,2^{nR_2}$.
A mismatched decoder for the cognitive MAC is defined by a mapping 
of the form (\ref{eq: qn additive MAC})
where the decoder outputs the message $(i,j)$ iff for all $(i',j')\neq (i,j)$
\begin{flalign}
q_n(\bx_1(i),\bx_2(i,j),\by)> q_n(\bx_1(i'),\bx_2(i',j'),\by).\label{eq: decision rule cognitive mac}
\end{flalign}
The capacity region of the cognitive mismatched MAC is defined similarly to that of the mismatched MAC.

Denote by $W_1,W_2$ the random messages, and the corresponding outputs of the decoder $\hat{W}_1,\hat{W}_2$. 
It is said that $E\geq 0$ is an achievable error exponent for the MAC if there exists a sequence of codebooks $\calC_{n,1},\calC_{n,2}$, $n\geq 1$ of rates $R_1$ and $R_2$, respectively, such that the average probability of error, $\bar{P}_{e,n}=\mbox{Pr}\{(W_1,W_2)\neq (W_1,W_2)\}$, that is incurred by the decoder $q_n$ when applied to codebooks $\calC_{n,1},\calC_{n,2}$ and the channel output satisfies $\liminf_{n\rightarrow\infty}-\frac{1}{n}\log \bar{P}_{e,n}\geq E$.

Two achievability schemes tailored for the mismatched cognitive MAC are presented next. 
The first encoding scheme is based on constant composition superposition coding, and the second on constant composition random binning. 

{ \bf Codebook Generation of User 1:} 
The codebook of the non-cognitive user is drawn the same way in both  coding methods. 
Fix a distribution $P_{X_1,X_2}\in\calP_n(\calX_1\times\calX_2)$. The codebook of user $1$ is composed of $2^{nR_1}$ codewords $\{\bx_1(i)\}, i=1,...,2^{nR_1}$ drawn  independently, each uniformly over the type-class $T(P_{X_1})$.  

{ \bf Codebook Generation of User 2 and Encoding:}
\begin{itemize} 
\item {\bf Superposition coding}: For each $\bx_1(i)$, user $2$ draws $2^{nR_2}$ codewords $\bx_2(i,j), j=1,...,2^{nR_2}$ conditionally independent given $\bx_1(i)$ uniformly over the conditional type-class $T(P_{X_1,X_2}|\bx_1(i))$. To transmit message $m_1$, encoder $1$ transmits $\bx_1(m_1)$. To transmit message $m_2$, encoder $2$, which is cognizant of the first user's message $m_1$, transmits $\bx_2(m_1,m_2)$.
\item
{\bf Random binning:} User $2$ draws $2^{n(R_2+\gamma)}$ codewords independently, each uniformly over $T(P_{X_2})$ and partitions them into $2^{nR_2}$ bins, i.e., 
$\{\bx_2[k,j]\}$, $k=1,...,2^{n\gamma}$, $j=1,...,2^{nR_2}$.
The quantity $\gamma$ is given by 
\begin{flalign}
\gamma=I_P(X_1;X_2)+\epsilon
\end{flalign}
for an arbitrarily small $\epsilon>0$. 
To transmit message $m_1$, encoder $1$ transmits $\bx_1(m_1)$. To transmit message $m_2$, encoder $2$, which is cognizant of the first user's message $m_1$, looks for a codeword in the $m_2$-th bin, $\bx_2[k,m_2]$ such that $(\bx_1(m_1),\bx_2[k,m_2])\in T(P_{X_1,X_2})$. If more than one such $k$ exists, the encoder chooses one of them arbitrarily, otherwise an error is declared. Thus, the encoding of user $2$ defines a mapping from the pairs of messages $(m_1,m_2)$ to a transmitted codeword $\bx_2$, which is denoted by $\bx_2(m_1,m_2)$, in parentheses, as opposed to the square brackets of $\bx_2[k,m_2]$.
 \end{itemize}
\noindent{ \bf Decoding:} The decoder chooses $(i,j)$ such that $q(\bx_1(i),\bx_2(i,j),\by)$ is maximal according to (\ref{eq: decision rule cognitive mac}), where ties are regarded as errors. 

\vspace{0.2cm}

The resulting achievable error-exponents for the mismatched cognitive MAC using superposition coding are presented next.
Let 
\begin{flalign}\label{eq: Pe1 Pe2 dfn}
 \bar{P}_{e,1}^{sup}= &\mbox{Pr}\left\{\hat{W}_1\neq W_1\right\}\nonumber\\
  \bar{P}_{e,2}^{sup}=&\mbox{Pr}\left\{\hat{W}_1=W_1,\hat{W_2}\neq W_2\right\}
  \end{flalign}  
  when superposition coding is employed.

   Let $Q\in\calP(\calX_1\times\calX_2\times\calY)$ be given. Define the following sets of p.m.f.'s that will be useful in what follows:
\begin{flalign}\label{eq: sets definitions}
&\calK(Q) \triangleq\{f\in\calP(\calX_1\times\calX_2\times\calY):\; f_{X_1,X_2}=Q_{X_1,X_2}\}\nonumber\\
&\calG_q(Q) \triangleq \{f\in\calK(Q):\; \bE_f\{q(X_1,X_2,Y)\}\geq \bE_Q\{q(X_1,X_2,Y)\} \}\nonumber\\
&\calL_1(Q) \triangleq  \{f\in \calG_q(Q) :\; f_{X_2,Y}=Q_{X_2,Y}\} \nonumber\\
&\calL_2(Q) \triangleq  \{f\in \calG_q(Q) :\; f_{X_1,Y}=Q_{X_1,Y}\} \nonumber\\
&\calL_0(Q) \triangleq  \{f\in \calG_q(Q) :\; f_{Y}=Q_{Y}\} .
\end{flalign}
 
\begin{theorem}\label{th: achievable error exponents II cognitive mismatch II} 
Let $P=P_{X_1,X_2}P_{Y|X_1,X_2}$, then \begin{flalign}
 \bar{P}_{e,2}^{sup} \doteq & 2^{-nE_2(P,R_2)}\label{eq: P_e2 expression a} \\
  \bar{P}_{e,1}^{sup} \doteq & 2^{-nE_1(P,R_1,R_2)}\label{eq: P_e2 expression}
 \end{flalign}
 where
 \begin{flalign}\label{eq: exponents dfn}
 &E_2(P,R_2)\nonumber\\
 =&\min_{P'\in \calK(P)}\bigg[ D(P'\| P)+ \min_{\tilde{P}\in \calL_2(P')} \left|I_{\tilde{P}}(X_2;Y|X_1)-R_2\right|^+ \bigg ],\nonumber\\
  &E_1(P,R_1,R_2)\nonumber\\
 =&\min_{P'\in \calK(P)}\bigg[ D(P'\| P)
 + \min_{\tilde{P}\in \calL_0(P')}  \left| I_{\tilde{P}}(X_1;Y)+| I_{\tilde{P}}(X_2;Y|X_1)-R_2|^+ -R_1\right|^+\bigg ].
\end{flalign}
\end{theorem}
The proof of Theorem \ref{th: achievable error exponents II cognitive mismatch II} can be found in Appendix \ref{sc: Outline of the proof of Theorem 2}.
We note that Theorem \ref{th: achievable error exponents II cognitive mismatch II} implies that 
\begin{flalign}
E_{sup}(P,R_1,R_2)=\min\left\{E_2(P,R_2),E_1(P,R_1,R_2) \right\}
\end{flalign}
is the error exponent induced by the superposition coding scheme.

\noindent Define the following functions
\begin{flalign}\label{eq: primes dfn}
 R_1'(P)\triangleq& \min_{\tilde{P}\in \calL_1(P)} I_{\tilde{P}}(X_1;Y,X_2)\nonumber\\
R_2'(P)\triangleq& \min_{\tilde{P}\in \calL_2(P)} I_{\tilde{P}}(X_2;Y|X_1)\nonumber\\
R_1''(P,R_2)\triangleq& \min_{\tilde{P}\in \calL_0(P)}  I_{\tilde{P}}(X_1;Y)+ |I_{\tilde{P}}(X_2;Y|X_1)-R_2|^+ \nonumber\\
R_2''(P,R_1) \triangleq& \min_{\tilde{P}\in \calL_0(P)}  \bigg\{I_{\tilde{P}}(X_2;Y)- I_{\tilde{P}}(X_2;X_1)
 + |I_{\tilde{P}}(X_1;Y,X_2)-R_1|^+\bigg\}
.\end{flalign}

Note that for $E_2(P,R_2),E_1(P,R_1,R_2)$ to be zero we must have $P'=P$. Theorem \ref{th: achievable error exponents II cognitive mismatch II}  therefore implies that the following region is achievable:
\begin{flalign}\label{eq: sup constraints}
\calR_{cog}^{sup}(P)=\left\{ (R_1,R_2):\;\begin{array}{l}
R_2\leq R_2'(P), \\
R_1\leq R_1''(P,R_2) 
\end{array}\right\}.
\end{flalign}

Consider the following region\footnote{Note that for convenience, in the first inequality of (\ref{eq: alternative representation sup MAC}) $ \min_{\tilde{P}\in \calL_2(P)} I_{\tilde{P}}(X_2;Y|X_1)$ is written explicitly instead of the
abbreviated notation $R_2'(P)$.}:
\begin{flalign}
\tilde{\calR}_{cog}^{sup}(P)=\nonumber\\
\bigg\{ (R_1,R_2):\;
&\begin{array}{cc}
R_2\leq   \min_{\tilde{P}\in \calL_2(P)} I_{\tilde{P}}(X_2;Y|X_1), \\
R_1+R_2\leq \min_{\tilde{P}\in\calL_0^{sup}(P)} I(X_1,X_2;Y) 
\end{array}\bigg\},\label{eq: alternative representation sup MAC}
\end{flalign}
where \begin{flalign}
\calL_0^{sup}(P)=&\bigg\{\tilde{P}\in \calL_0(P):\;   I_{\tilde{P}}(X_1;Y)\leq R_1 \bigg\}.
\end{flalign}

The following theorem provides a random coding converse using superposition coding, and also implies the equivalence of $\calR_{cog}^{sup}(P)$ and $\tilde{\calR}_{cog}^{sup}(P)$.
\begin{theorem}\label{cr: first equivalence} (Random Coding Converse for Superposition Coding)
If $(R_1,R_2)\notin\tilde{\calR}_{cog}^{sup}(P_{X_1,X_2,Y})$ then
the average probability of error, averaged over the ensemble
of random codebooks drawn according to $P_{X_1,X_2}$ using superposition coding,
approaches one as the blocklength tends to infinity.
\end{theorem}
The proof of Theorem \ref{cr: first equivalence} appears in Appendix \ref{sc: first equivalence}. It follows similarly to the proof of Theorem 3 of \cite{Lapidoth96}
\begin{corollary}\label{cr: first corollary equivalence superposition}
\begin{flalign}
\calR_{cog}^{sup}(P)=\tilde{\calR}_{cog}^{sup}(P).
\end{flalign}
\end{corollary}
The inclusion $\calR_{cog}^{sup}(P)\subseteq \tilde{\calR}_{cog}^{sup}(P)$, follows from Theorem \ref{cr: first equivalence} and since $ \calR_{cog}^{sup}(P)$ is an achievable region. The proof of the opposite direction $\tilde{\calR}_{cog}^{sup}(P)\subseteq \calR_{cog}^{sup}(P)$ appears in Appendix \ref{sc: proof of Corollary equivalence 1}.

Since by definition the capacity region is a closed convex set, this yields the following achievability theorem. 
\begin{theorem}\label{th: achievable II cognitive mismatch II}
The capacity region of the finite alphabet cognitive MAC $P_{Y|X_1,X_2}$ with decoding metric $q(x_1,x_2,y)$ contains 
the set of rate-pairs 
\begin{flalign}
\calR_{cog}^{sup}=\mbox{ closure of CH of }\underset{P}{\cup} \tilde{\calR}_{cog}^{sup}(P)
\end{flalign}
where the union is over all $P\in\calP(\calX_1\times\calX_2\times\calY)$ with conditional $P_{Y|X_1,X_2}$ given by the channel. 
\end{theorem}

The error exponents achievable by random binning are presented next. 
Let  $\bar{P}_{e,2}^{bin},\bar{P}_{e,1}^{bin}$ be defined as follows
\begin{flalign}\label{eq: Pe1 Pe2 binning dfn}
 \bar{P}_{e,1}^{bin}= &\mbox{Pr}\left\{\hat{W}_1\neq W_1\right\}\nonumber\\
  \bar{P}_{e,2}^{bin}=&\mbox{Pr}\left\{\hat{W}_1=W_1,\hat{W_2}\neq W_2\right\}
  \end{flalign}  
when random binning is employed. 
\begin{theorem}\label{th: achievable error exponents II cognitive mismatch II binning} 
Let $P=P_{X_1,X_2}P_{Y|X_1,X_2}$, then \begin{flalign}\
 \bar{P}_{e,2}^{bin} \doteq & 2^{-nE_2(P,R_2)}\nonumber\\
  \bar{P}_{e,1}^{bin} \doteq & 2^{-n\min\{E_0(P,R_1,R_2),E_1(P,R_1)\}}
  \end{flalign}
 where 
 \begin{flalign}\label{eq: E0bPR1R1 dfn}
  &E_1(P,R_1)\nonumber\\
 &=\min_{P'\in\calK(P)}\bigg[ D(P'\| P)+  \min_{\tilde{P}\in \calL_1(P')}  | I_{\tilde{P}}(X_1;Y,X_2)-R_1|^+ \bigg ]\nonumber\\
&E_0(P,R_1,R_2)=\max\{E_1(P,R_1,R_2), E_{0,b}(P,R_1,R_2)\}
  \end{flalign}
and where 
 \begin{flalign}
&E_{0,b}(P,R_1,R_2)\nonumber\\
 =&\min_{P'\in\calK(P)}\bigg[ D(P'\| P)+  \min_{\tilde{P}\in \calL_0(P')} \bigg | I_{\tilde{P}}(X_2;Y)-I_P(X_1;X_2)\nonumber\\
 & \quad\quad\quad \quad\quad\quad +| I_{\tilde{P}}(X_1;Y,X_2)-R_1|^+ -R_2\bigg|^+\bigg ],
\end{flalign}
with $E_2(P,R_2)$ and $E_1(P,R_1,R_2)$ defined in (\ref{eq: exponents dfn}).
\end{theorem}
The proof of Theorem \ref{th: achievable error exponents II cognitive mismatch II binning}  appears in Appendix \ref{sc:Proof of Theorem 3}. The derivation of  the exponent associated with $\bar{P}_{e,1}^{bin}$ makes use of \cite[Lemma 3]{ScarlettFabregas2012}, where
achievable error exponents for the non-cognitive MAC obtained by a constant-composition random coding are characterized. 
We note that Theorem \ref{th: achievable error exponents II cognitive mismatch II binning}  implies that 
\begin{flalign}
E_{bin}(P,R_1,R_2)=\min\left\{E_2(P,R_2),E_0(P,R_1,R_2),E_1(P,R_1) \right\}
\end{flalign}
is the error exponent induced by the random binning scheme. 
Theorem \ref{th: achievable error exponents II cognitive mismatch II binning} also implies that for fixed $P=P_{X_1,X_2}P_{Y|X_1,X_2}$, the following rate-region is achievable: 
\begin{flalign}\label{eq: bin constraints}
\calR_{cog}^{bin}(P)=\left\{ (R_1,R_2):\;\begin{array}{l}
R_1\leq R_1'(P),\\
R_2\leq R_2'(P), \\
R_1\leq R_1''(P,R_2)\mbox{ or }  R_2 \leq R_2''(P,R_1)
\end{array}\right\},
\end{flalign}
where $R_1'(P),R_2'(P),R_1''(P,R_2),R_2''(P,R_1)$ are defined in (\ref{eq: primes dfn}).
Next it is proven that $\calR_{cog}^{bin}(P)$ has the following alternative expression. Consider the rate-region:
\begin{flalign}\label{cr: equivalent representation bin}
&\tilde{\calR}_{cog}^{bin}(P)=\nonumber\\
&\left\{ (R_1,R_2):\;
\begin{array}{l}
R_1\leq  \min_{\tilde{P}\in \calL_1(P)} I_{\tilde{P}}(X_1;Y,X_2),\\
R_2\leq   \min_{\tilde{P}\in \calL_2(P)} I_{\tilde{P}}(X_2;Y|X_1), \\
R_1+R_2\leq \min_{\tilde{P}\in \calL_0^{bin}(P)}I_{\tilde{P}}(X_1,X_2;Y)
\end{array}\right\}
\end{flalign}
where \begin{flalign}
\calL_0^{bin}(P)=&\bigg\{\tilde{P}\in \calL_0(P):\;  \nonumber\\
& I_{\tilde{P}}(X_1;Y)\leq R_1, I_{\tilde{P}}(X_2;Y)- I_{P}(X_2;X_1)\leq R_2 \bigg\}
\end{flalign}
The following theorem provides a random coding converse for random binning and it also implies that $\tilde{\calR}_{cog}^{bin}(P)$ is an alternative expression for $ \calR_{cog}^{bin}(P)$.
\begin{theorem}\label{cr: second equivalence} (Random Coding Converse for Random Binning)
If $(R_1,R_2)\notin \tilde{\calR}_{cog}^{bin}(P_{X_1,X_2,Y})$ then
the average probability of error, averaged over the ensemble
of random codebooks drawn according to $P_{X_1,X_2}$ using binning, 
approaches one as the blocklength tends to infinity.
\end{theorem}
The proof of Theorem \ref{cr: second equivalence} adheres closely to that of Theorem \ref{cr: first equivalence} and is thus omitted. 
\begin{corollary}\label{cr: Cor cor corollary}
\begin{flalign}
\calR_{cog}^{bin}(P)=\tilde{\calR}_{cog}^{bin}(P).
\end{flalign}
\end{corollary}
The inclusion $\calR_{cog}^{bin}(P)\subseteq \tilde{\calR}_{cog}^{bin}(P)$, follows from Theorem \ref{cr: second equivalence} and since $ \calR_{cog}^{bin}(P)$ is an achievable region. The proof of the opposite direction $\tilde{\calR}_{cog}^{bin}(P)\subseteq \calR_{cog}^{bin}(P)$ appears in Appendix \ref{sc: proof of Corollary equivalence 2}.

Note that in fact $R_{cog}^{bin}(P)$ can be potentially enlarged as follows: 
\begin{lemma}\label{th: enlargement Theorem}
Let $(R_1,R_2)\in \calR_{cog}^{bin}(P)$ then $(R_1+R_2,0)$ is also achievable by random binning. 
\end{lemma}
\begin{proof}
The lemma follows since the cognitive encoder can assign some of the information it transmits to the non-cognitive user. The message $W_1$ can be split to $W_{1,a}$ and $W_{1,b}$ corresponding to rates $R_{1,a}$ and $R_{2,b}$. User $1$ transmits $W_{1,a}$ and user $2$ transmits $(W_2,W_{1,b})$. The achievable rate-region becomes \begin{flalign}\label{eq: dfn calR star cog}
\calR_{cog}^{bin,*}(P)=\bigg\{ (R_1,R_2):\; \exists R_{1,a},R_{1,b}\geq 0:\; 
&R_1=R_{1,a}+R_{1,b},\nonumber\\
&R_{1,a}\leq R_1'(P),\nonumber\\
&R_{1,b}+R_2\leq R_2'(P), \nonumber\\
&R_{1,a}\leq R_1''(P,R_{1,b}+R_2)\mbox{ or }  R_{1,b}+R_2 \leq R_2''(P,R_{1,a})\bigg\}.
\end{flalign}
\end{proof} 
The resulting region of rates achievable by random binning is described in the following theorem. 
\begin{theorem}\label{th: achievable region binning}
The capacity region of the finite alphabet cognitive MAC $P_{Y|X_1,X_2}$ with decoding metric $q(x_1,x_2,y)$ contains 
the set of rate-pairs 
\begin{flalign}\label{eq: convex hull bin}
\calR_{cog}^{bin}=\mbox{ closure of CH of }\underset{P}{\cup} \calR_{cog}^{bin,*}(P)
\end{flalign}
where the union is over all $P\in\calP(\calX_1\times\calX_2\times\calY)$ with conditional $P_{Y|X_1,X_2}$ given by the channel. 
\end{theorem}

\section{Discussion - Mismatched Cognitive MAC}\label{sc: discussion}

A few comments are in order:
\begin{itemize}
\item
Next a Lemma is proved, which establishes the fact that the achievable region of the cognitive mismatched MAC, $\calR_{cog}^{bin}$, contains the achievable region of the mismatched MAC, $\calR_{LM}$ (\ref{eq: calRLM_dfn}).
\begin{lemma}\label{lm: Lemma yoffi}
$
\calR_{LM}\subseteq \calR_{cog}^{bin}
$.
\end{lemma}
\begin{proof}

For this proof we use the expression (\ref{cr: equivalent representation bin}). 
Recall the definition of $\tilde{R}_2$ in (\ref{eq: calRLM_dfn}) it satisfies
\begin{flalign}
\tilde{R}_2
=&\min_{f\in \calD_{(2)}} I_f(X_2;Y|X_1)+I_f(X_1;X_2) \nonumber\\
= &\min_{f\in \calD_{(2)}} D(f_{X_1,X_2,Y}\| P_{X_2}P_{X_1,Y})
\nonumber\\
\leq &\min_{ f\in \calD_{(2)}:\; f_{X_1,X_2}=P_{X_1}P_{X_2}} D(f_{X_1,X_2,Y}\| P_{X_2}P_{X_1,Y})\nonumber\\
= &\min_{ f\in \calL_2\left(P_{X_1}P_{X_2}P_{Y|X_1,X_2}\right)} I_f(X_2;Y|X_1)\nonumber\\
= &R_2'(P_{X_1}P_{X_2}P_{Y|X_1,X_2}).
\end{flalign}
where $R_2'(P)$ is defined in (\ref{eq: primes dfn}). Similarly, 
\begin{flalign}
 \tilde{R}_0
=&\min_{f\in \calD_{(0)}} I_f(X_1,X_2;Y)+I_f(X_1;X_2) \nonumber\\
= &\min_{f\in \calD_{(0)}} D(f_{X_1,X_2,Y}\| P_{X_2}P_{X_1}P_{Y})
\nonumber\\
\leq & \min_{ f\in \calD_{(0)}:\; f_{X_1,X_2}=P_{X_1}P_{X_2}} D(f_{X_1,X_2,Y}\| P_{X_2}P_{X_1}P_{Y})
\nonumber\\
\leq &\min_{ f\in \calL_0^{bin}(P_{X_1}P_{X_2}P_{Y|X_1,X_2})}  I_f(X_1,X_2;Y),
\nonumber\\
\end{flalign}
where the inequality follows since $ \calL_0^{bin}(P_{X_1}P_{X_2}P_{Y|X_1,X_2})\subseteq \calD_{(0)}$. A similar inequality can be derived for $\tilde{R}_1$ and is omitted.

The definition of $\calR_{cog}^{bin}$ (\ref{eq: convex hull bin}) includes a union over all $P_{X_1,X_2}$ 
including product p.m.f.'s of the form $P_{X_1,X_2}=P_{X_1}P_{X_2}$, whereas the definition of $\calR_{LM}$ (\ref{eq: calRLM_dfn}) includes a union over product p.m.f.'s alone, and thus $\calR_{LM}\subseteq \calR_{cog}^{bin}$.
\end{proof}
The fact that $\calR_{LM}\subseteq \calR_{cog}^{bin}$ is not surprising as one expects that an achievable region of a mismatched cognitive MAC should be larger than that of a  mismatched (non-cognitive) MAC. 
\item
Next it is proved that 
in the matched case $\calR_{cog}^{sup}=\calR_{cog}^{bin}$ and  the regions are both equal to the matched capacity region.
\begin{proposition}\label{th: matched theorem }
In the matched case
where
\begin{flalign}
q(x_1,x_2,y)=\log p(y|x_1,x_2), 
\end{flalign}
 $\calR_{cog}^{sup}=\calR_{cog}^{bin}=\calR_{cog}^{match}$ where
\begin{flalign}\label{eq: calR_COG_matched_case superposition}
\calR_{cog}^{match}=&
 \underset{P_{X_1,X_2}}{\cup}  \big\{(R_1,R_2):\nonumber\\
R_2&\leq I_P(X_2;Y|X_1),
R_1+R_2\leq I_P(X_1,X_2;Y)\big\},
\end{flalign}
and $P$ abbreviates $P_{X_1,X_2}P_{Y|X_1,X_2}$. 
\end{proposition}
\begin{proof} 
The proof that $\calR_{cog}^{sup}=\calR_{cog}^{match}$ follows very similarly to the proof of \cite[Proposition $1$]{Lapidoth96}, and is thus omitted. 
To prove $\calR_{cog}^{bin}=\calR_{cog}^{match}$, note that as in \cite[Proposition $1$]{Lapidoth96}, for every $\tilde{P}\in \calL_1(P)$, we have
\begin{flalign}
I_{\tilde{P}}(X_1;Y,X_2)& =I_{\tilde{P}}(X_1;Y|X_2)+I_P(X_1;X_2)\nonumber\\
& =H_P(Y|X_2)- H_{\tilde{P}}(Y|X_1,X_2)+I_P(X_1;X_2)\nonumber\\
& \stackrel{(a)}{\geq} H_P(Y|X_2)- E_{\tilde{P}}\log( P_{Y|X_1,X_2})+I_P(X_1;X_2)\nonumber\\
& \stackrel{(b)}{\geq} H_P(Y|X_2)- E_P\log( P_{Y|X_1,X_2})+I_P(X_1;X_2)\nonumber\\
& = I_P(X_1;Y,X_2),
\end{flalign}
where $(a)$ follows from the non-negativity of the divergence and $(b)$ follows since $\tilde{P}\in \calL_2(P)$ and thus $E_P \log(q)\leq E_{\tilde{P}}\log(q)$. 
Similarly, one can show that for every $\tilde{P}\in \calL_2(P)$
\begin{flalign}
I_{\tilde{P}}(X_2;Y|X_1) \geq I_P(X_2;Y,X_1),
\end{flalign}
and that for all $\tilde{P}\in \calL_0(P)$
\begin{flalign}
I_{\tilde{P}}(X_1,X_2;Y) \geq I_P(X_1,X_2;Y).
\end{flalign}
This yields that the union of rate-pairs $(R_1,R_2)$ achievable by binning contains the rate-pairs satisfying
\begin{flalign}
R_1\leq &I_P(X_1;Y,X_2)\nonumber\\
R_2\leq &I_P(X_2;Y|X_1)\nonumber\\
R_1+R_2\leq &I_P(X_1,X_2;Y)
 .
\end{flalign}
But, since $I_P(X_1;Y,X_2)\geq I_P(X_1;Y)$, the sum of the bounds on the individual rates is looser than the sum-rate bound, are we get the achievable vertex point $(R_1,R_2)=\left(I_P(X_1;Y), I_P(X_2;Y|X_1)\right)$, and by enlarging this region according to Lemma \ref{th: enlargement Theorem} combined with time-sharing, the region 
\begin{flalign}
R_2\leq &I_P(X_2;Y|X_1)\nonumber\\
R_1+R_2\leq &I_P(X_1,X_2;Y)
 ,
\end{flalign}
is achievable using random binning and is contained in $\calR_{cog}^{bin,*}(P)$ and thus also in $\calR_{cog}^{bin}$. This implies that $\calR_{cog}^{match}\subseteq \calR_{cog}^{bin}$, and hence $\calR_{cog}^{match}=\calR_{cog}^{bin}$.
\end{proof}
Theorem \ref{th: matched theorem } is clearly an example for which $\calR_{LM}\subseteq\calR_{cog}^{sup}=\calR_{cog}^{bin}$ with obvious cases in which the inclusion is strict. 

\item Next we compare the achievable regions using superposition coding and random binning, i.e., $\calR_{cog}^{sup}$ (\ref{eq: bin constraints}) and $\calR_{cog}^{bin,*}$ (defined in Lemma \ref{th: enlargement Theorem}).
 In principle, neither region $\calR_{cog}^{sup}(P)$, $\calR_{cog}^{bin,*}(P)$ dominates the other, 
 as the second inequality in (\ref{eq: sup constraints}) is stricter than the third inequality (\ref{eq: bin constraints}) and the first inequality in (\ref{eq: bin constraints}) does not appear in (\ref{eq: sup constraints}).

It is easily verified that unless $ R_1''(P,R_2'(P))> R_1'(P)$ and $R_2''(P,R_1'(P))> R_2'(P)$
 (that is, unless the sum of the individual rates bounds is stricter than both sum-rate bounds), we have $  \calR_{cog}^{sup}(P)\subseteq \calR_{cog}^{bin,*}(P)$, otherwise, the opposite inclusion $  \calR_{cog}^{bin,*}(P)\subseteq \calR_{cog}^{sup}(P)$ may occur.

\item An example is presented next for which $\calR_{cog}^{sup}\subset \calR_{cog}^{bin}$ with strict inclusion. 
Consider the following parallel MAC which is a special case of the channel that was studied by Lapidoth \cite[Section IV, Example 2]{Lapidoth96}. Let the alphabets $\calX_1=\calX_2=\calY_1=\calY_2$ be binary $\{0,1\}$. The output of the channel is given by
\begin{flalign}
Y=& (Y_1,Y_2)
\end{flalign}
where
\begin{flalign}
Y_1=& X_1\nonumber\\
Y_2=& X_2\oplus Z,
\end{flalign}
$\oplus$ denotes modulo-$2$ addition, and $Z\sim \mbox{Bernoulli}(p'')$, with the decoding metric
\begin{flalign}
q(x_1,x_2,(y_1,y_2))=& -\frac{1}{2}(x_1\oplus y_1 +x_2\oplus y_2).
\end{flalign}
Now, the capacity region of this channel with non-cognitive users was established by Lapidoth and is given by the rectangle
\begin{flalign}
\calR_{LM}=\left\{ (R_1,R_2):\; \begin{array}{l} 
R_2\leq  1-h_2(p'')\\
R_1\leq 1
\end{array}\right\}
\end{flalign}
where $h_2(p'')=-p''\log(p'')-(1-p'')\log(1-p'')$. 

From Lemma \ref{lm: Lemma yoffi} we know that $\calR_{LM}\subseteq \calR_{cog}^{bin}$. Consequently, from Lemma \ref{th: enlargement Theorem}, we obtain that $\calR_{cog}^{bin}$ contains the region
\begin{flalign}\label{eq: trapezoid region parallel}
\calR_{cog}^{bin'}
=\left\{ (R_1,R_2):\; \begin{array}{l} 
R_2\leq  1-h_2(p'')\\
R_1+R_1\leq  2-h_2(p'')\end{array}\right\},
\end{flalign}
since this is also the capacity region of the matched cognitive MAC, it can be concluded that random binning combined with enlargement of $R_1$ according to Lemma \ref{th: enlargement Theorem} achieves the capacity region.

Next, it is demonstrated that the vertex point $(R_1,R_2)=(1,1-h_2(p''))$ is not achievable by superposition coding when the cognitive user is user $2$.
Consider the sum-rate bound in $\tilde{\calR}_{cog}^{sup}(P)$ (\ref{eq: alternative representation sup MAC})
\begin{flalign}
R_1+R_2\leq &\max_{P_{X_1,X_2}}\min_{\tilde{P}\in \calL_0^{sup}(P)}I(X_1,X_2;Y_1,Y_2)\nonumber\\
=& \max_{P_{X_1,X_2}}\min_{\tilde{P}\in \calL_0(P): I_{\tilde{P}}(X_1;Y)\leq R_1}I(X_1,X_2;Y_1,Y_2).
\end{flalign}
Consequently, for $R_1=1$, we obtain
\begin{flalign}
R_2\leq & \max_{P_{X_1,X_2}}\min_{\tilde{P}\in \calL_0(P): I_{\tilde{P}}(X_1;Y)\leq 1}I(X_1,X_2;Y_1,Y_2)-1\nonumber\\
=& \max_{P_{X_1,X_2}}\min_{\tilde{P}\in \calL_0(P)}I(X_1,X_2;Y_1,Y_2)-1.
\end{flalign}
since clearly $ I_{\tilde{P}}(X_1;Y)\leq 1$ is always satisfied as $X_1$ is binary. Now, the term $\max_{P_{X_1,X_2}}\min_{\tilde{P}\in \calL_0(P)}I(X_1,X_2;Y_1,Y_2)$ is simply the maximal achievable rate by ordinary random coding for the single-user channel from $X=(X_1,X_2)$ to $Y=(Y_1,Y_2)$, which was characterized for this channel in \cite[Section IV, Example 2]{Lapidoth96}, and is given by $2\left(1-h_2(p''/2)\right)$. This yields that if $R_1=1$ then $R_2$ achievable by superposition coding satisfies 
\begin{flalign}
R_2\leq & 1-2h_2(p''/2),
\end{flalign}
which is strictly lower than $R_2=1-h_2(p'')$ which is achievable by binning.

It should be noted that although $\calR_{cog}^{sup}\subset  \calR_{cog}^{bin}$ in this case, if the roles of the users were reversed, i.e., user $1$ were cognitive and user $2$ were non-cognitive, the vertex point $(R_1,R_2)=(1,1-h_2(p''))$ would have been achievable by superposition as well (see the explanation following Theorem \ref{th: single user inner bound}).

\item
 The following theorem provides a condition for $\calR_{cog}^{sup}(P)\subseteq \calR_{cog}^{bin}(P)$.
\begin{proposition}\label{cr: bin eq sup}
If $R_2'(P)\geq I_P(X_2;Y)-I_P(X_1;X_2)$ then $\calR_{cog}^{sup}(P)\subseteq \calR_{cog}^{bin}(P)$. \end{proposition}
\begin{proof}
It is argued that if $R_2'(P)\geq I_P(X_2;Y)-I_P(X_1;X_2)$, the constraint $R_1\leq R_1'(P) $ in (\ref{eq: bin constraints}) is looser than $R_1\leq R_1''(P,R_2)$. To realize this, let 
$R_2=I_P(X_2;Y)-I_P(X_1;X_2)+\Delta$ where $\Delta\geq0$. We have
\begin{flalign}
R_1''(P,R_2)=&\min_{\tilde{P}\in \calL_0(P)}  I_{\tilde{P}}(X_1;Y)+| I_{\tilde{P}}(X_2;Y|X_1)-R_2|^+\nonumber\\
\leq&\min_{\tilde{P}\in \calL_0(P)}  I_{\tilde{P}}(X_1;Y)+| I_{\tilde{P}}(X_2;Y|X_1)-R_2+\Delta|^+\nonumber\\
\stackrel{(a)}{\leq} &\min_{\tilde{P}\in \calL_1(P)}  I_{\tilde{P}}(X_1;Y)+| I_{\tilde{P}}(X_2;Y|X_1)-R_2+\Delta|^+\nonumber\\
\stackrel{(b)}{=}& \min_{\tilde{P}\in \calL_1(P)}  I_{\tilde{P}}(X_1;Y,X_2) \nonumber\\
=& R_1'(P)
\end{flalign}
where $(a)$ follows since $ \calL_1(P) \subseteq \calL_0(P)$, and $(b)$ follows since for all $\tilde{P}\in   \calL_1(P)$, 
\begin{flalign}
&I_{\tilde{P}}(X_2;Y|X_1) \nonumber\\
=&I_{\tilde{P}}(X_2;Y)-I_{\tilde{P}}(X_2;X_1)+I_{\tilde{P}}(X_2;X_1|Y)\nonumber\\
=&I_P(X_2;Y)-I_P(X_2;X_1)+I_{\tilde{P}}(X_2;X_1|Y)
\end{flalign}
\end{proof}
In fact, Theorem \ref{cr: bin eq sup} generalizes the fact that binning performs as well as superposition coding in the matched case, since in the matched case one has $R_2'(P)\geq I_P(X_2;Y)-I_P(X_1;X_2)$.

\end{itemize}

\section{The Mismatched Single-User Channel}\label{sc: The Implication for the Single-User Channel}

This section shows that achievable rates for the mismatched single-user DMC can be derived from the maximal sum-rate of an appropriately chosen mismatched cognitive MAC. 

Similar to the definitions in  \cite{Lapidoth96}, consider the single-user mismatched DMC $P_{Y|X}$ with input alphabet $\calX$ and decoding metric $q(x,y)$. Let $\calX_1$ and $\calX_2$ be finite alphabets and let $\phi$ be a given mapping
$\phi:\; \calX_1\times \calX_2\rightarrow \calX$. We will study the rate-region of the mismatched cognitive MAC with input alphabets $\calX_1,\calX_2$ and output alphabet $\calY$, whose input-output relation is given by 
\begin{flalign}\label{eq: induced MAC}
P_{Y|X_1,X_2}(y|x_1,x_2)&=P_{Y|X}(y|\phi(x_1,x_2)),
\end{flalign}
where the right hand side is the probability of the output of the single-user channel to be $y$ given that its input is $\phi(x_1,x_2)$. The decoding metric $q(x_1,x_2,y)$ of the mismatched cognitive MAC is defined in terms of that of the single-user channel: 
\begin{flalign}
q(x_1,x_2,y)&=q(\phi(x_1,x_2),y).
\end{flalign}
The resulting mismatched cognitive MAC will be referred to as the cognitive MAC induced by the single-user channel, or more specifically, induced by $\left(P_{Y|X},q(x,y),\calX_1,\calX_2,\phi\right)$.
Note that, in fact, $X_1$ and $X_2$ can be regarded as auxiliary random variables for the original single-user channel $P_{Y|X}$. 

In  \cite[Theorem 4]{Lapidoth96}, it is shown that the mismatch capacity of the single-user channel
 is lower-bounded by $R_1+R_2$ for any pair $(R_1,R_2)$ that, for some mapping
$\phi$ and for some distributions $P_{X_1}$ and $P_{X_2}$ satisfy
\begin{flalign}\label{eq: single user Lapidoth}
R_1<&\min_{\underset{f_{X_1,X_2}=f_{X_1}f_{X_2}}{f\in \calD_{(1)}} } I_f(X_1;Y|X_2)\nonumber\\
R_2<&\min_{\underset{f_{X_1,X_2}=f_{X_1}f_{X_2}}{f\in \calD_{(2)}} } I_f(X_2;Y|X_1)\nonumber\\
R_1+R_2<&\min_{\underset{f_{X_1,X_2}=f_{X_1}f_{X_2}}{f\in \calD_{(0)}} } I_f(X_1,X_2;Y),
\end{flalign}
where $\calD_{(i)},i=0,1,2$ are defined in (\ref{eq: calD dfn}) and $P_{X_1,X_2,Y}=P_{X_1}P_{X_2}P_{Y|\phi(X_1,X_2)}$. 

The proof of \cite[Theorem 4]{Lapidoth96} is based on expurgating the product codebook of the MAC (\ref{eq: induced MAC}) containing $2^{n(R_1+R_2)}$ codewords $\bv(i,j)=(\bx_1(i),\bx_2(j))$ and only keeping the $\bv(i,j)$'s that are composed of $\bx_1(i),\bx_2(j)$ which are jointly $\epsilon$-typical with respect to the product p.m.f.\ $P_{X_1}P_{X_2}$. It is shown that the expurgation causes a negligible loss of rate and therefore makes is possible to consider minimization over product measures in (\ref{eq: single user Lapidoth}). 

It is easy to realize that in the cognitive MAC case as well, if $(R_1,R_2)$ is an achievable rate-pair for the induced cognitive MAC, $R_1+R_2$ is an achievable rate for the inducing single-user channel. 
While the users of the non-cognitive induced MAC of \cite{Lapidoth96} exercise a limited degree of cooperation (by expurgating the appropriate codewords of the product codebook), the induced cognitive MAC introduced here enables a much higher degree of cooperation between users. 
There is no need to expurgate codewords for cases of either
 superposition coding or random binning, since the codebook generation guarantees that for all $(i,j)$, $\left(\bx_1(i),\bx_2(i,j)\right)$ lies in the desired joint type-class $T(P_{X_1,X_2})$.

Let $\calR_{cog}(P)=\calR_{cog}^{sup}(P)\cup\calR_{cog}^{bin,*}(P)$. 
For convenience the dependence of $\calR_{cog}(P)$ on $q$ is made explicit and is denoted by $\calR_{cog}(P_{X_1,X_2,Y},q(x_1,x_2,y))$, \begin{theorem}\label{th: single user inner bound}
For all finite $\calX_1,\calX_2$, $P_{X_1,X_2}\in\calP(\calX_1,\calX_2)$ and $\phi:\calX_1\times\calX_2\rightarrow\calX$, 
the capacity of the single-user mismatched DMC $P_{Y|X}$ with decoding metric $q(x,y)$ is lower bounded by the sum-rate resulting from $\calR_{cog}(P_{X_1,X_2}P_{Y|\phi(X_1,X_2)},q(\phi(x_1,x_2),y))$. 
\end{theorem}
While Theorem \ref{th: single user inner bound} may not improve the rate (\ref{eq: single user Lapidoth}) achieved in \cite[Theorem 4]{Lapidoth96} in optimizing over all $\left(\calX_1,\calX_2,\phi,P_{X_1},P_{X_2}\right)$, it can certainly improve the achieved rates for given $\left(\phi,P_{X_1},P_{X_2}\right)$ as demonstrated in Section \ref{sc: discussion}, and thereby may reduce the computational complexity required to find a good code. 

Next, it is demonstrated how superposition coding can be used to achieve the rate of Theorem \ref{th: single user inner bound}. 
Consider the region of rate-pairs $(R_1,R_2)$ which satisfy
\begin{flalign}\label{eq: sup constraints single user reversed}
R_1&\leq r_1(P)\triangleq \min_{\tilde{P}\in \calL_1(P)} I_{\tilde{P}}(X_1;Y|X_2)\nonumber\\
R_2&\leq r_2(P,R_1)\triangleq \min_{\tilde{P}\in \calL_0(P)}  I_{\tilde{P}}(X_2;Y)+ |I_{\tilde{P}}(X_1;Y|X_2)-R_1|^+.
\end{flalign}
This region is obtained by reversing the roles of the users in (\ref{eq: sup constraints}), i.e., setting user $1$ as the cognitive one. 

We next show that the sum-rate resulting from $\calR_{cog}^{bin}(P)$ (\ref{eq: bin constraints}) is upper-bounded by the sum-rate resulting from the union of the regions (\ref{eq: sup constraints single user reversed}) and (\ref{eq: sup constraints}) that are achievable by superposition coding. To verify this, note that (\ref{eq: bin constraints}) is contained in the union of rate-regions
\begin{flalign}\label{eq: bin constraints single user}
\left\{ (R_1,R_2):\;    \begin{array}{l}
R_2\leq R_2'(P), \\
R_1\leq R_1''(P,R_2)\end{array}\right\}
\cup
\left\{ (R_1,R_2):\;    \begin{array}{l}
R_1\leq R_1'(P),\\
 R_2 \leq R_2''(P,R_1)
\end{array}\right\}, 
\end{flalign}
and the region on the l.h.s.\ is equal to $\calR_{cog}^{sup}$ (\ref{eq: sup constraints}). Now, let $(R_1^*,R_2^*)$ be the vertex point of the region on the r.h.s.\, i.e., the point which satisfies
\begin{flalign}
R_1^*&= \min_{\tilde{P}\in \calL_1(P)} I_{\tilde{P}}(X_1;Y,X_2)\nonumber\\
R_2^*&= \min_{\tilde{P}\in \calL_0(P)}  \bigg[I_{\tilde{P}}(X_2;Y)- I_{\tilde{P}}(X_2;X_1) + |I_{\tilde{P}}(X_1;Y,X_2)-R_1^*|^+\bigg].
\end{flalign}
By definition of $ \calL_1(P)$ and $ \calL_0(P)$, denoting $R_1^{**}\triangleq R_1^*-I_P(X_2;X_1)$ and $R_2^{**}\triangleq R_2^*+I_P(X_2;X_1)$ this yields
\begin{flalign}
R_1^{**}&= \min_{\tilde{P}\in \calL_1(P)} I_{\tilde{P}}(X_1;Y|X_2)\nonumber\\
R_2^{**}&= \min_{\tilde{P}\in \calL_0(P)}  \bigg[I_{\tilde{P}}(X_2;Y)+ |I_{\tilde{P}}(X_1;Y|X_2)-R_1^{**}|^+\bigg].
\end{flalign}
Since clearly $(R_1^{**},R_2^{**})$ lies in (\ref{eq: sup constraints single user reversed}) and since $R_1^{**}+R_2^{**}=R_1^*+R_2^*$, we obtain that the sum-rate resulting from 
(\ref{eq: sup constraints single user reversed}) is equal to that of the r.h.s.\ of (\ref{eq: bin constraints single user}).
Consequently, the sum-rate that can be achieved by the union of the two superposition coding schemes (with user $1$ the cognitive and user $2$ the non-cognitive and vice versa) is an upper bound on the sum-rate achievable by binning, and in Theorem \ref{th: single user inner bound}, one can replace $\calR_{cog}(P_{X_1,X_2}P_{Y|\phi(X_1,X_2)},q(\phi(x_1,x_2),y))$ with the union of the regions (\ref{eq: sup constraints single user reversed}) and (\ref{eq: sup constraints}). This yields the following corollary.
\begin{corollary}
For all finite $\calX_1,\calX_2$, $P_{X_1,X_2}\in\calP(\calX_1,\calX_2)$ and $\phi:\calX_1\times\calX_2\rightarrow\calX$, the capacity of the single-user mismatched DMC $P_{Y|X}$ with decoding metric $q(x,y)$ is lower bounded by the rate
\begin{flalign}
\max\left\{ R_2'(P)+R_1''(P, R_2'(P)) , r_1(P)+r_2(P,r_1(P))  \right\}, 
\end{flalign}
which is achievable by superposition coding, where $P=P_{X_1,X_2}\times P_{Y|\phi(X_1,X_2)}$, and the functions $r_1(P), r_2(P,R_1)$ are defined in (\ref{eq: sup constraints single user reversed}). 
\end{corollary}

\section{Conclusion}\label{eq: Conclusion}

In this paper, two encoding methods for cognitive multiple access channels: a superposition coding scheme and a random binning scheme were analyzed. Tight single-letter expressions were obtained for the resulting error exponents of these schemes. The achievable regions were characterized and proofs were provided for the random coding converse theorems. While apparently neither of the schemes dominates the other, there are certain conditions under which each of the schemes is not inferior (in terms of reliably transmitted rates) compared to the other scheme for a given random coding distribution $P_{X_1,X_2}$. An example was also discussed for a cognitive MAC whose achievable region using random binning is strictly larger than that obtained by superposition coding. The matched case was also studied, in which the achievable regions of both superposition coding and random binning are equal to the capacity region, which is often strictly larger than the achievable region of the non-cognitive matched MAC. 

In certain cases, binning is more advantageous than superposition coding in terms of memory requirements: superposition coding requires the cognitive user to use a separate codebook for every possible codeword of the non-cognitive user, i.e., a collection of $2^{n(R_1+R_2)}$ codewords. Binning on the other hand allows encoder $2$ to decrease memory requirements to $2^{nR_1}+2^{n(R_2+I_P(X_1;X_2))}$ codewords at the cost of increased encoding complexity\footnote{In this case, the choice of $\bx_2(i,j)$ should not be arbitrary among the vectors that are jointly typical with $\bx_1(i)$ in the $j$-th bin, but rather, a deterministic rule, e.g., pick the jointly typical $\bx_2[k,j]$ with the lowest $k$.}. 

The achievability results were further specialized to obtain a lower bound on the mismatch capacity of the single-user channel by investigating a cognitive multiple access channel whose achievable sum-rate serves as a lower bound on the single-user channel's capacity. This generalizes Lapidoth's scheme \cite{Lapidoth96} for a single-user channel that is based on the non-cognitive MAC. 
While the users of the non-cognitive MAC of \cite{Lapidoth96} exercise a limited degree of cooperation (by expurgating the appropriate codewords of the product codebook), the  cognitive MAC introduced here allows for a much higher degree of cooperation between users. Neither superposition coding nor random binning requires expurgating codewords, since the codebook generation guarantees that all pairs of transmitted codewords lie in the desired joint type-class $T(P_{X_1,X_2})$. 
Additionally, Lapidoth's lower bound for the single-user channel requires optimizing over the parameters (auxiliary random variables alphabets and distributions), but
when the full optimization over the parameters is infeasible, the bound provided in this paper can be strictly larger and may reduce the computational complexity required to find a good code. We further show that by considering the two superposition schemes (with user $1$ being the cognitive and user $2$ the non-cognitive and with reversed roles) one can achieve a sum-rate that is at least as high as that of the random binning scheme.

\section{Acknowledgement}
The author would like to thank the anonymous reviewers of the International Symposium on Information Theory 2013 for their very helpful comments and suggestions, which helped improving this paper.

\appendix\label{sc: Proofs of the main results}
Throughout the proofs the method of types is used. For a survey of the method the reader is referred to \cite[Chapter 11.1]{CoverThomas2006}
. In particular inequalities involving type sizes, such as if $\hat{P}_{\bx,\by}= Q$ then, 
\begin{flalign}\label{eq: typesize}
c_n^{-1}2^{nH_Q(X|Y)}\leq |T(Q_{X,Y}|\by)|\leq 2^{nH_Q(X|Y)}, 
\end{flalign}
where $c_n=(n+1)^{(|\calX||\calY|-1)}$, i.e., $|T(Q_{X,Y}|\by)|\doteq 2^{nH_Q(X|Y)}$. Additionally, if $A$ is an event that can be expressed as a union over type-classes of $\calX$, and $\bX$ is a random $n$ vector over $\calX^n$, since the number of types grows polynomially with $n$, we have
\begin{flalign}\label{eq: Sanov}
&\max_{\tilde{P}\in A} \mbox{Pr}\left(\bX\in T(\tilde{P})\right)\leq  \mbox{Pr}\left(\bX\in A\right)
\leq c_n\max_{\tilde{P}\in  A}  \mbox{Pr}\left(\bX\in T(\tilde{P})\right),
\end{flalign}
i.e., $\mbox{Pr}\left(\bX\in A\right)\doteq \max_{\tilde{P}\in A}  \mbox{Pr}\left(\bX\in T(\tilde{P})\right)$. Finally, if $\bY$ is conditionally i.i.d.\ given a deterministic vector $\bx$ with $P(Y_i=y|X_i=x)\sim Q_{Y=y|X=x}$, then for $\hat{P}_X=\hat{P}_{\bx}$
\begin{flalign}\label{eq: memoryless channel output type}
Q^n\left( \bY\in T(\hat{P}_{X,Y}|\bx) |\bX=\bx\right) \doteq 2^{-nD(\hat{P}\| Q|\hat{P}_{\bx})},
\end{flalign}
where
\begin{flalign}
D(\hat{P}\| Q|\hat{P}_{\bx})&= \sum_{x,y} \hat{P}_{X,Y}(x,y)\log\frac{\hat{P}(y|x)}{Q(y|x)}.
\end{flalign}

Recall the definitions of the sets of p.m.f.'s (\ref{eq: sets definitions}). Next, we define similar sets of empirical p.m.f.'s. 
 Let $Q\in\calP_n(\calX_1\times\calX_2\times\calY)$ be given. Define the following sets of p.m.f.'s that will be useful in what follows
\begin{flalign}
&\calK_n(Q) \triangleq\{f\in\calP_n(\calX_1\times\calX_2\times\calY):\; f_{X_1,X_2}=Q_{X_1,X_2}\} \nonumber\\
&\calG_{q,n}(Q) \triangleq \{f\in\calK_n(Q):\; \bE_f\{q(X_1,X_2,Y)\}\geq \bE_Q\{q(X_1,X_2,Y)\} \}\nonumber\\
&\calL_{1,n}(Q) \triangleq  \{f\in \calG_{q,n}(Q) :\; f_{X_2,Y}=Q_{X_2,Y}\} \label{eq: empirical definitions 3}\\
&\calL_{2,n}(Q) \triangleq  \{f\in \calG_{q,n}(Q) :\; f_{X_1,Y}=Q_{X_1,Y}\}\label{eq: empirical definitions 4}\\
&\calL_{0,n}(Q) \triangleq  \{f\in \calG_{q,n}(Q) :\; f_{Y}=Q_{Y}\} .\label{eq: empirical definitions 5}
\end{flalign}

\subsection{Proof of Theorem \ref{th: achievable error exponents II cognitive mismatch II} }\label{sc: Outline of the proof of Theorem 2}

Assume without loss of generality that the transmitted messages are $(m_1,m_2)=(1,1)$, and let $W^n$ stand for $P_{Y|X}^n$.
We have
\begin{flalign}
\bar{P}_{e,2}^{sup} =&\sum_{(\bx_1,\bx_2)\in T(P_{X_1,X_2}),\by}\frac{W^n(\by|\bx_1,\bx_2)}{|T(P_{X_1,X_2})|}e(\hat{P}_{\bx_1,\bx_2,\by})\label{eq: eq50}
\end{flalign}
where\footnote{We denote $e(\hat{P}_{\bx_1,\bx_2,\by})$ rather than $e(\bx_1,\bx_2,\by)$ because it is easily verified that the average error probability conditioned on $(\bx_1,\bx_2,\by)$ is a function of the joint empirical measure $\hat{P}_{\bx_1,\bx_2,\by}$.} $e(\hat{P}_{\bx_1,\bx_2,\by})$ is the average probability of $\{\hat{W}_1=1,\hat{W}_2\neq 1\}$ given $\{\bX_1,\bX_2,\bY)=(\bx_1,\bx_2,\by)\}$. 
Since the decoder successfully decodes $W_2$ only if $q_n(\bx_1(1),\bx_2(1),\by)> q_n(\bx_1(1),\bx_2(j),\by)$ for all $j\neq 1$ and since $e(\hat{P}_{\bx_1,\bx_2,\by})$ can be regarded as the probability of at least one "success" in $M_2-1$ Bernoulli trials, we have 
\begin{flalign}
e(\hat{P}_{\bx_1,\bx_2,\by})=&1-\left[1-a(\hat{P}_{\bx_1,\bx_2,\by})\right]^{M_2-1},
\end{flalign}
where
\begin{flalign}\label{eq: eq44}
a(\hat{P}_{\bx_1,\bx_2,\by})=\sum_{\bx_2'\in T(P_{X_1,X_2}|\bx_1):\; \frac{q_n(\bx_1,\bx_2',\by)}{q_n(\bx_1,\bx_2,\by)}\geq 1}\frac{1}{|T(P_{X_2|X_1})|},
\end{flalign}
which is the probability that $\bX_2'$ drawn uniformly over $T(P_{X_1,X_2}|\bx_1)$ will yield a higher metric than the transmitted codeword, i.e., $q_n(\bx_1,\bX_2',\by)\geq q_n(\bx_1,\bx_2,\by)$.
From Lemma 1 in \cite{SomekhBaruchMerhav07} we know that for $a\in[0,1]$, one has 
\begin{flalign}\label{eq: doteq Lemma}
\frac{1}{2}\min\{1,Ma\}\leq 1-\left[1-a\right]^M\leq \min\{1,Ma\}.
\end{flalign}
Consequently,
\begin{flalign}
e(\hat{P}_{\bx_1,\bx_2,\by})\doteq &\min\left \{1,M_2a(\hat{P}_{\bx_1,\bx_2,\by})\right\}\nonumber\\
=& 2^{-n\left| -\frac{1}{n}\log a(\hat{P}_{\bx_1,\bx_2,\by})-R_2\right|^+}.\label{eq: eq52}
\end{flalign}
Now, by noting that the summation in (\ref{eq: eq44}) is over conditional type-classes of $\bx_2'$ given $(\bx_1,\by)$, such that $\bE_{\hat{P}_{\bx_1,\bx_2',\by}}\{q(X_1,X_2,Y)\}\geq  \bE_{\hat{P}_{\bx_1,\bx_2,\by}}\{q(X_1,X_2,Y)\}$ and $\hat{P}_{\bx_1,\bx_2}=\hat{P}_{\bx_1,\bx_2'}$ as $\bx_2,\bx_2'$ are both drawn conditionally independent given $\bx_1$, and using (\ref{eq: typesize}) and (\ref{eq: Sanov}) we get
\begin{flalign}
&a(\hat{P}_{\bx_1,\bx_2,\by})\nonumber\\
\doteq &\max_{\tilde{P}\in \calL_{2,n}(\hat{P}_{\bx_1,\bx_2,\by}) }\frac{| T(\tilde{P}_{X_2|X_1,Y})|}{|T(P_{X_2|X_1})|}\nonumber\\
\doteq & 2^{-n\min_{\tilde{P}\in \calL_{2,n}(\hat{P}_{\bx_1,\bx_2,\by})} I_{\tilde{P}}(X_2;Y|X_1) },\label{eq: eq54}
\end{flalign}
where $ \calL_{2,n}(\hat{P}_{\bx_1,\bx_2,\by})$ is defined in (\ref{eq: empirical definitions 4}).
Using (\ref{eq: typesize}) and (\ref{eq: memoryless channel output type}) and gathering (\ref{eq: eq50}), (\ref{eq: eq52}) and (\ref{eq: eq54}) this yields 
\begin{flalign}
 \bar{P}_{e,2}^{sup} \doteq &2^{
 -n\min_{P'\in\calK_n(P)}\left(D(P'\| P)+\min_{\tilde{P}\in \calL_{2,n}(P')}\left| I_{\tilde{P}}(X_2;Y|X_1) -R_2\right|^+\right) } \nonumber\\
 \doteq &2^{-nE_2(P,R_2)}
 \end{flalign}
where the last step follows since by continuity of $\left| I_{\tilde{P}}(X_2;Y|X_1) -R_2\right|^+$ in $\tilde{P}$ for sufficiently large $n$ we can replace the minimization over empirical p.m.f.'s $\calL_{2,n}(P')$ with a minimization over $\calL_{2}(P')$, and in a similar manner, we can replace the minimization over $\calK_n(P)$ with a minimization over $\calK(P) $ and obtain exponentially equivalent expressions.
 This concludes the proof of (\ref{eq: P_e2 expression a}).
 
Next, to compute the average error probability of the event of erroneously decoding $W_1$, i.e., $\hat{W}_1\neq 1$, we have
\begin{flalign}
\bar{P}_{e,1}^{sup} =&\sum_{(\bx_1,\bx_2)\in T(P_{X_1,X_2}),\by}\frac{W^n(\by|\bx_1,\bx_2)}{|T(P_{X_1,X_2})|}\nu(\hat{P}_{\bx_1,\bx_2,\by})\label{eq: still valid for binning 2}
\end{flalign}
where 
\begin{flalign}\label{eq: still valid}
&\nu(\hat{P}_{\bx_1,\bx_2,\by})=\sum_{T(\tilde{P}_{X_1,X_2,Y}|\by): \frac{q_n(\bx_1,\bx_2',\by)}{q_n(\bx_1,\bx_2,\by)}\geq 1, \tilde{P}_{X_1,X_2}=P_{X_1,X_2}} \nonumber\\
&\times \mbox{Pr}\left\{\exists i\neq 1,j:(\bX_1(i),\bX_2(i,j))\in T(\tilde{P}_{X_1,X_2,Y}|\by)\right\}.
\end{flalign}
Now, fix some $m\in\{2,...,M_1\}$ and note that $\nu(\bx_1,\bx_2,\by)$ can be regarded as the probability of at least one "success" in $M_1-1$ Bernoulli trials, i.e., 
\begin{flalign}\label{eq: Bernoulli trials}
&  \mbox{Pr}\left\{\exists i\neq 1,j:(\bX_1(i),\bX_2(i,j))\in T(\tilde{P}_{X_1,X_2,Y}|\by)\right\}\nonumber\\
=& 1-\left[\mbox{Pr}\left\{\not \exists j:\; (\bX_1(m),\bX_2(m,j))\in T(\tilde{P}_{X_1,X_2,Y}|\by)\right\} \right]^{M_1-1}.
\end{flalign}
The event $\left\{\not \exists j:\; (\bX_1(m),\bX_2(m,j))\in T(\tilde{P}_{X_1,X_2,Y}|\by)\right\}$ is the union of two disjoint events
\begin{itemize}
\item $A\triangleq\{ \bX_1(m)\notin T(\tilde{P}_{X_1,Y}|\by) \}$
\item $B\triangleq \bigg\{\bX_1(m)\in T(\tilde{P}_{X_1,Y}|\by)$ and  
$\not\exists j:\; (\bX_1(m),\bX_2(m,j))\in T(\tilde{P}_{X_1,X_2,Y}|\by)\bigg\}$.
\end{itemize}
Since $\bX_1(m)$ is drawn uniformly over $T(P_{X_1})=T(\tilde{P}_{X_1})$ we have from (\ref{eq: typesize})
\begin{flalign}
\mbox{Pr}\left\{A\right\}&= 1-\frac{|T(\tilde{P}_{X_1,Y}|\by)|}{|T(\tilde{P}_{X_1})|}\doteq 1-2^{-nI_{\tilde{P}}(X_1;Y)}
\end{flalign}
and 
\begin{flalign}
\mbox{Pr}\left\{B\right\}&= \frac{T(\tilde{P}_{X_1,Y}|\by)}{|T(\tilde{P}_{X_1})|}\cdot\left[ 1-\frac{|T(\tilde{P}_{X_1,X_2,Y}|\bx_1,\by)|}{|T(P_{X_1,X_2}|\bx_1)|} \right]^{M_2-1}\nonumber\\
&\doteq 2^{-nI_{\tilde{P}}(X_1;Y)}\cdot\left[ 1-2^{-nI_{\tilde{P}}(X_2;Y|X_1) } \right]^{M_2}.
\end{flalign}
Therefore, since the events are disjoint
\begin{flalign}\label{eq: added equation for brevity}
\mbox{Pr}\left\{A\cup B\right\}&= \mbox{Pr}\left\{A\right\}+\mbox{Pr}\left\{B\right\}   \nonumber\\
&\doteq 1-2^{-nI_{\tilde{P}}(X_1;Y)}+2^{-nI_{\tilde{P}}(X_1;Y)}\cdot\left[ 1-2^{-nI_{\tilde{P}}(X_2;Y|X_1) } \right]^{M_2} \nonumber\\
&= 1-2^{-nI_{\tilde{P}}(X_1;Y)}\left[1-\left[ 1-2^{-nI_{\tilde{P}}(X_2;Y|X_1) } \right]^{M_2}\right] \nonumber\\
&\doteq 1-2^{-nI_{\tilde{P}}(X_1;Y)+| I_{\tilde{P}}(X_2;Y|X_1)-R_2|^+ },
\end{flalign}
where the last step follows from (\ref{eq: doteq Lemma}) applied to $a=2^{-nI_{\tilde{P}}(X_2;Y|X_1) }$, and from (\ref{eq: Bernoulli trials}) we have 
\begin{flalign}
&  \mbox{Pr}\left\{\exists i\neq 1,j:(\bX_1(i),\bX_2(i,j))\in T(\tilde{P}_{X_1,X_2|Y}|\by)\right\} \nonumber\\
\doteq&  1-\left[1- 2^{-nI_{\tilde{P}}(X_1;Y)+| I_{\tilde{P}}(X_2;Y|X_1)-R_2|^+ }\right]^{M_1}\label{eq: this is the equation i need}\\
\doteq&2^{-n\left| I_{\tilde{P}}(X_1;Y)+| I_{\tilde{P}}(X_2;Y|X_1)-R_2|^+ -R_1\right|^+},\label{eq: end of citation}
\end{flalign}
which follows from (\ref{eq: doteq Lemma}) applied to 
\noindent $a=2^{-nI_{\tilde{P}}(X_1;Y)+| I_{\tilde{P}}(X_2;Y|X_1)-R_2|^+ }$.

Now, by noting that the summation in (\ref{eq: still valid}) is over conditional type-classes of $(\bx_2',\bx_1')$ given $\by$, such that $\bE_{\hat{P}_{\bx_1',\bx_2',\by}}\{q(X_1,X_2,Y)\}\geq  \bE_{\hat{P}_{\bx_1,\bx_2,\by}}\{q(X_1,X_2,Y)\}$ 
 and $\hat{P}_{\bx_1',\bx_2'}=\hat{P}_{\bx_1,\bx_2}$, and using (\ref{eq: Sanov}) we get
\begin{flalign}
&\nu(\hat{P}_{\bx_1,\bx_2,\by})\nonumber\\
\doteq &\max_{\tilde{P}\in \calL_{0,n}(\hat{P}_{\bx_1,\bx_2,\by}) } 2^{-n\left| I_{\tilde{P}}(X_1;Y)+| I_{\tilde{P}}(X_2;Y|X_1)-R_2|^+ -R_1\right|^+}\label{eq: eq5adkjbvkjbv}
\end{flalign}
where $ \calL_{0,n}(\hat{P}_{\bx_1,\bx_2,\by})$ is defined in (\ref{eq: empirical definitions 5}). 

Using (\ref{eq: memoryless channel output type}) and gathering (\ref{eq: still valid for binning 2}), (\ref{eq: eq5adkjbvkjbv}) this yields 
\begin{flalign}\label{eq: last continuity step}
 \bar{P}_{e,1}^{sup} \doteq & 2^{-n\min_{P'\in \calK_n(P)} \left(D(P'\|P)+\min_{\tilde{P}\in \calL_{0,n}(P')}  \left| I_{\tilde{P}}(X_1;Y)+| I_{\tilde{P}}(X_2;Y|X_1)-R_2|^+ -R_1\right|^+ \right)}\nonumber\\
 \doteq & 2^{-nE_1(P,R_1,R_2)}.
 \end{flalign}
where the last step follows since by continuity of $\left| I_{\tilde{P}}(X_1;Y)+| I_{\tilde{P}}(X_2;Y|X_1)-R_2|^+ -R_1\right|^+$ in $\tilde{P}$, for sufficiently large $n$ we can replace the minimization over empirical p.m.f.'s $ \calL_{0,n}(P')$ with a minimization over $ \calL_{0}(P')$, and similarly, we can replace the minimization over $\calK_n(P)$ with a minimization over $\calK(P) $ and obtain exponentially equivalent expressions.

\subsection{Proof of Theorem \ref{th: achievable error exponents II cognitive mismatch II binning}}\label{sc:Proof of Theorem 3}

The first observation similarly to \cite{Gelfand-Pinsker-1980}, is that the probability that for given $(m_1,m_2)$, no $k\in \{1,...,2^{n\gamma}\}$ exists such that $(\bx_1(m_1),\bx_2[k,m_2])\in T(P_{X_1,X_2})$ vanishes super-exponentially fast provided that $\gamma\geq I_P(X_1;X_2)+\epsilon$ for some $\epsilon>0$ since
\begin{flalign}\label{eq: GP property}
&\mbox{Pr}\left\{\not\exists k:\; (\bX_1(m_1),\bX_2[k,m_2])\in T(P_{X_1,X_2}) \right\}\nonumber\\
= &\left(1-\frac{ |T(P_{X_2|X_1})| }{|T(P_{X_2})|}\right)^{2^{n\gamma}}\nonumber\\
\doteq &\left(1-2^{-nI_P(X_1;X_2)}\right)^{2^{n\gamma}}\nonumber\\
\leq &\exp\{-2^{n(\gamma-I_P(X_1;X_2))}\}.
\end{flalign}
Moreover, from the union bound over $(m_1,m_2)$ we have that the probability that there exists $(m_1,m_2)$ such that $ (\bx_1(m_1),\bx_2[k,m_2])\notin T(P_{X_1,X_2})$ for all $k$ vanishes super exponentially fast, therefore, we have 
\begin{flalign}\label{eq: GP observation}
\mbox{Pr}\left\{\bX_2(m_1,m_2)=\tilde{\bx}_2 |\bX_1(m_1)=\tilde{\bx}_1\right\}\doteq \frac{1\{ \tilde{\bx}_2\in T(P_{X_1,X_2} | \tilde{\bx}_1)\}}{|T(P_{X_1,X_2} | \tilde{\bx}_1)|},
\end{flalign}
i.e., uniform in the appropriate conditional type-class.

Assume without loss of generality that $\bx_2(1,1)=\bx_2[1,1]$. 
As a direct consequence of  (\ref{eq: GP observation}), $\bar{P}_{e,2}^{bin}$ can be calculated similarly to its calculation for the superposition coding scheme yielding 
\begin{flalign}
\bar{P}_{e,2}^{bin} \doteq  2^{-nE_2(P,R_2)}.
\end{flalign}
 The calculation of $\bar{P}_{e,1}^{bin}$ on the other hand, differs from that of the superposition coding since the last step of (\ref{eq: Bernoulli trials}) is no longer valid as it does not correspond to $2^{nR_1}-1$ independent Bernoulli experiments. Hence,
 \begin{flalign}
\bar{P}_{e,1}^{bin} =&\sum_{(\bx_1,\bx_2)\in T(P_{X_1,X_2}),\by}\frac{W^n(\by|\bx_1,\bx_2)}{|T(P_{X_1,X_2})|}\zeta(\hat{P}_{\bx_1,\bx_2,\by})\label{eq: still valid for binning}
\end{flalign}
where 
\begin{flalign}\label{eq: zetaequation}
&\zeta(\hat{P}_{\bx_1,\bx_2,\by})=\sum_{T(\tilde{P}_{X_1,X_2,Y}|\by): \frac{q_n(\bx_1,\bx_2',\by)}{q_n(\bx_1,\bx_2,\by)}\geq 1, \tilde{P}_{X_1,X_2}=P_{X_1,X_2}} \nonumber\\
&\times \mbox{Pr}\left\{
\exists i\neq 1, (k,j):(\bX_1(i),\bX_2[k,j])\in T(\tilde{P}_{X_1,X_2|Y}|\by)
\right\}\end{flalign}
where $i\in\{2,...,2^{nR_1}\}$, $j\in\{1,...,2^{nR_2}\}$, $k\in\{1,...,2^{n\gamma}\}$. 

We distinguish between two cases: $(k,j)\neq (1,1)$ and $(k,j)=(1,1)$ since 
\begin{flalign}
& \mbox{Pr}\left\{
\exists i\neq 1, (k,j):(\bX_1(i),\bX_2[k,j])\in T(\tilde{P}_{X_1,X_2|Y}|\by)\right\} \nonumber\\
\doteq & \mbox{Pr}\left\{
\exists i\neq 1:(\bX_1(i),\bX_2[1,1])\in T(\tilde{P}_{X_1,X_2|Y}|\by)\right\} \nonumber\\
&+ \mbox{Pr}\left\{
\exists i\neq 1, (k,j)\neq (1,1):(\bX_1(i),\bX_2[k,j])\in T(\tilde{P}_{X_1,X_2|Y}|\by)\right\} .
\end{flalign}

\noindent{\bf case a: $(k,j)\neq (1,1)$:} 
In this case we use Lemma 3 of  \cite{ScarlettFabregas2012} to obtain:
\begin{flalign}
&\mbox{Pr}\left\{
\exists i\neq 1, (k,j)\neq (1,1):(\bX_1(i),\bX_2[k,j])\in T(\tilde{P}_{X_1,X_2,Y}|\by)\right\}\nonumber\\
\doteq & 2^{-n\max\{\psi_1(\tilde{P},R_1,R_2), \psi_2(\tilde{P},R_1,R_2)\}},
\end{flalign}
where
\begin{flalign}
\psi_1(\tilde{P},R_1,R_2)\triangleq & \left|I_{\tilde{P}}(X_1;Y)+\left| I_{\tilde{P}}(X_2;Y,X_1)-R_2-\gamma \right|^+-R_1 \right|^+\nonumber\\
\psi_2(\tilde{P},R_1,R_2)\triangleq &  \left|I_{\tilde{P}}(X_2;Y)+\left| I_{\tilde{P}}(X_1;Y,X_2)-R_1\right|^+-R_2 -\gamma\right|^+.
\end{flalign}
Since $\gamma=I_P(X_1,X_2)+\epsilon$ where $\epsilon$ is arbitrarily small\footnote{In fact, $\epsilon$ can be replaced, for example, with $\epsilon_n=n^{-1/2}$ to guarantee that the r.h.s.\ of (\ref{eq: GP property}) vanishes super-exponentially fast.}, the functions $\psi_1(\tilde{P},R_1,R_2)$ and $\psi_1(\tilde{P},R_1,R_2)$ converge to
\begin{flalign}
\tilde{\psi}_1(\tilde{P},R_1,R_2)\triangleq & \left|I_{\tilde{P}}(X_1;Y)+\left| I_{\tilde{P}}(X_2;Y|X_1)-R_2\right|^+-R_1 \right|^+\nonumber\\
\tilde{\psi}_2(\tilde{P},R_1,R_2)\triangleq &  \left|I_{\tilde{P}}(X_2;Y)-I_{\tilde{P}}(X_2;X_1)+\left| I_{\tilde{P}}(X_1;Y,X_2)-R_1\right|^+-R_2 \right|^+,
\end{flalign}
respectively, as $n$ tends to infinity.

\noindent{\bf case b: $(k,j)= (1,1)$:} In this case we have
\begin{flalign}
& \mbox{Pr}\left\{
\exists i\neq 1:(\bX_1(i),\bX_2[1,1])\in T(\tilde{P}_{X_1,X_2,Y}|\by)
\right\}\nonumber\\
=&1-\left[1-\frac{|T(\tilde{P}_{X_1|Y,X_2}|\by)|}{|T(P_{X_1})|}\right]^{M_1-1}\nonumber\\
\stackrel{(a)}{\doteq} &\min\{1,M_12^{-nI_{\tilde{P}}(X_1;Y,X_2)}\}\nonumber\\
=&2^{-n\left| I_{\tilde{P}}(X_1;Y,X_2)-R_1  \right|^+}.
\end{flalign}
where $(a)$ follows from (\ref{eq: doteq Lemma}). Gathering the two cases we obtain
\begin{flalign}
& \mbox{Pr}\left\{
\exists i\neq 1(j,k):(\bX_1(i),\bX_2[j,k])\in T(\tilde{P}_{X_1,X_2,Y}|\by)
\right\}\nonumber\\
\doteq &\max\left\{ 2^{-n\left| I_{\tilde{P}}(X_1;Y,X_2)-R_1  \right|^+} ,2^{-n\max\{\tilde{\psi}_1(\tilde{P},R_1,R_2), \tilde{\psi}_2(\tilde{P},R_1,R_2)\}}\right\}.
\end{flalign}

This yields similarly to (\ref{eq: last continuity step}) 
\begin{flalign}\
  \bar{P}_{e,1}^{bin} \doteq & 2^{-n\min\{E_0(P,R_1,R_2),E_1(P,R_1)\}}.
  \end{flalign}

\subsection{Proof of Theorem \ref{cr: first equivalence}}\label{sc: first equivalence}

The proof of the random coding converse follows the line of proof of Theorem 3 of \cite{Lapidoth96}. For the sake of completeness, the proof outline is repeated here and the differences between the proofs are reiterated: Recall that $P_{X_1,X_2}\in \calP_n(\calX_1\times\calX_2)$ is the random coding distribution. We need to show that
if the inequality 
\begin{flalign}\label{eq: Psi definition}
R_1+R_2\leq \min_{\tilde{P}\in\calL_0^{sup}(P)} I(X_1,X_2;Y) \triangleq \Psi(R_1,R_2,P)\end{flalign} 
is violated the ensemble average error probability tends to one as $n$ tends to infinity. The proof of the claim that the first inequality in (\ref{eq: alternative representation sup MAC}), i.e., 
$R_2\leq   \min_{\tilde{P}\in \calL_2(P)} I_{\tilde{P}}(X_2;Y|X_1)$ is a necessary condition for a vanishingly low average probability of error is simpler and thus omitted. 

We follow the steps in  \cite{Lapidoth96}:
\begin{itemize}
\item Step 1: if $R_1+R_2> \min_{\tilde{P}\in\calL_0^{sup}(P)} I(X_1,X_2;Y) $ then there exists a p.m.f.\ $\tilde{P}\in\calL_0^{sup}(P)$ such that
\begin{flalign}
R_1+R_2 >&  I_{\tilde{P}}(X_1,X_2;Y)\label{eq: eq1}\\
R_1>& I_{\tilde{P}}(X_1;Y)\label{eq: eq2}\\
\bE_{\tilde{P}}\{q\} >& \bE_P\{q\}\label{eq: eq3}
\end{flalign}
This follows directly as in \cite{Lapidoth96} by the convexity of the function $\Psi(R_1,R_2,P)$ defined in (\ref{eq: Psi definition}).

\item Step 2: By the continuity of the relative entropy functional
and by Step 1 
we can find some $\Delta>0$, some $\epsilon>0$ , and
a neighborhood $U$ of $\tilde{P}$ such that for every $f\in U$, 
and $\hat{P}(Y)\in T_n^{\epsilon}(P_Y)$ (where $T_n^{\epsilon}(P_Y)$ is the set of strongly $\epsilon$-typical sequences w.r.t.\ $P_Y$),
\begin{flalign}\label{ eq: eq65}
f_{X_1,X_2}=& P_{X_1,X_2}\nonumber\\
R_1+R_2> & D(f_{X_1,X_2|Y}\| P_{X_1}P_{X_2}|\hat{P}_y)+\Delta\nonumber\\
R_1> & D(f_{X_1|Y}\| P_{X_1}|\hat{P}_y)+\Delta\nonumber\\
\bE_f\{q\}>&\bE_P\{q\}+\Delta.
\end{flalign}
We next choose a sufficiently small neighborhood $V$ of
$P$ so that for every $\mu\in V$ we have
\begin{flalign}\label{eq: eq66}
\bE_\mu\{q\}<\bE_P\{q\}+\Delta,\nonumber\\
\mu_Y\in T_n^{\epsilon}(P_Y),
\end{flalign}
We can thus conclude that if the triple $(\bx_1(1),\bx_2(1,1),\by)$ has an empirical type in $V$, and if there exist codewords $(\bx_1(i),\bx_2(i,j))$ with $i\neq 1$ such that the empirical type of $(\bx_1(i),\bx_2(i,j),\by)$ is in $U$, then a decoding error must occur. 
\item Step 3: We will show that the probability that there is no $(i\neq 1,j)$ such that  $(\bx_1(i),\bx_2(i,j),\by)$ is in $U$ vanishes as $n$ tends to infinity. 
We can thus conclude that provided that $\by$ is $\epsilon$-typical w.r.t.\ $P_Y$,  the
conditional probability that there exist some $i\neq 1$ and $j$ 
such that the joint type of $(\bx_1(i),\bx_2(i,j),\by)$ is in $U$ approaches
one. In particular, by (\ref{ eq: eq65}), with probability approaching one,
there exists a pair of incorrect codewords that accumulate a
metric higher than $\bE_P\{q\}+\Delta$.

We have shown in (\ref{eq: this is the equation i need}) that for every $T(Q_{X_1,X_2|Y}|\by)$ such that the marginal distribution of $X_1,X_2$ induced by $Q_{Y,X_1,X_2}=\hat{P}_{\by}\times Q_{X_1,X_2|Y}$ is $P_{X_1,X_2}$, 
\begin{flalign}\label{eq: conditional probability given y}
& \mbox{Pr}\left\{\not\exists i\neq 1, j:\; (\bX_1(i),\bX_2(i,j))\in T(Q_{X_1,X_2|Y}|\by) \right\}\nonumber\\
\doteq& \left(1-2^{-I_Q(X_1;Y)}\min\left\{1,2^{-n[I_Q(X_2;Y|X_1)-R_2]}\right\}\right)^{M_1}\nonumber\\
= & \left(1-\min\left\{2^{-I_Q(X_1;Y)},2^{-n[I_Q(X_1,X_2;Y)-R_2]}\right\}\right)^{M_1}
\end{flalign}
where $M_1=2^{nR_1}$, and it is easily verified to be vanishing as $n$ tends to infinity if $R_1> I_{Q}(X_1;Y)$ and $R_1+R_2>I_{Q}(X_1,X_2;Y)$. 

\item Step 4: By the LLN, the probability that the joint type of
$(\bX_1(1),\bX_2(1,1),\bY)$ is in $V$ approaches one as the blocklength
tends to infinity. In this event, by (\ref{eq: eq66}), the correct codewords
accumulate a metric that is smaller than $\bE_P\{q\}+\Delta$, and thus by
(\ref{ eq: eq65}) an error is bound to occur if there exists such a codeword-pair. Since the conditional
probability (\ref{eq: conditional probability given y}) given $\by$ approaches one for any $\epsilon$-typical $\by$, and since all p.m.f.Õs in $V$ have marginals in
$T_\epsilon^n(P_Y)$ it follows that the probability of error approaches one as the
message length tends to infinity, and the theorem is proved.
\end{itemize}

\subsection{Proof of Corollary \ref{cr: first corollary equivalence superposition}}\label{sc: proof of Corollary equivalence 1}

It remains to prove that $\tilde{\calR}_{cog}^{sup}(P)\subseteq \calR_{cog}^{sup}(P)$. 
To realize this, fix $R_2$ and let $\tilde{R}_1$ be the corresponding maximal user $1$'s rate resulting from $ \calR_{cog}^{sup}(P)$ (\ref{eq: sup constraints}), i.e., the rate $\tilde{R}_1$ which satisfies
\begin{flalign}\label{eq: eq 9328987685746}
\tilde{R}_1 =& \min_{\tilde{P}\in \calL_0(P)}  I_{\tilde{P}}(X_1;Y)+ |I_{\tilde{P}}(X_2;Y|X_1)-R_2|^+.
\end{flalign}
Now, observe that since $|t|^+\geq t, \forall t$, we have 
\begin{flalign}
\tilde{R}_1 =\min_{\tilde{P}\in \calL_0(P):\;  I_{\tilde{P}}(X_1;Y)\leq \tilde{R}_1}  I_{\tilde{P}}(X_1;Y)+ |I_{\tilde{P}}(X_2;Y|X_1)-R_2|^+.
\end{flalign}
This implies that $ \calR_{cog}^{sup}(P)$ is equivalent to 
\begin{flalign}\label{eq: 843279587987}
\bigg\{ (R_1,R_2):\;
&\begin{array}{cc}
R_2\leq   \min_{\tilde{P}\in \calL_2(P)} I_{\tilde{P}}(X_2;Y|X_1), \\
R_1\leq \min_{\tilde{P}\in\calL_0^{sup}(P)}  I_{\tilde{P}}(X_1;Y)+ |I_{\tilde{P}}(X_2;Y|X_1)-R_2|^+
\end{array}\bigg\},
\end{flalign}
which obviously contains $\tilde{\calR}_{cog}^{sup}(P)$ since $|t|^+\geq t, \forall t$, and implies that $\tilde{\calR}_{cog}^{sup}(P)\subseteq \calR_{cog}^{sup}(P)$. 

Another way to realize that $\tilde{\calR}_{cog}^{sup}(P)\subseteq \calR_{cog}^{sup}(P)$ is to consider $E_1(P,R_1,R_2)$ (\ref{eq: exponents dfn}) and note that whenever $R_1>  I_{\tilde{P}}(X_1;Y)$ we have $  \left| I_{\tilde{P}}(X_1;Y)+| I_{\tilde{P}}(X_2;Y|X_1)-R_2|^+ -R_1\right|^+>0$, and that $\left| I_{\tilde{P}}(X_2;Y|X_1)-R_2\right|^+$ $\geq I_{\tilde{P}}(X_2;Y|X_1)-R_2$.

\subsection{Proof of Corollary \ref{cr: Cor cor corollary}}\label{sc: proof of Corollary equivalence 2}

It remains to prove that $\tilde{\calR}_{cog}^{bin}(P)\subseteq \calR_{cog}^{bin}(P)$. It can be shown similarly to Corollary \ref{cr: first corollary equivalence superposition} (see (\ref{eq: eq 9328987685746})-(\ref{eq: 843279587987})), that the union of the following regions is equivalent to $\calR_{cog}^{bin}(P)$ (\ref{eq: bin constraints}):
\begin{flalign}
\left\{ (R_1,R_2):\;\begin{array}{l}
R_1\leq R_1'(P),\\
R_2\leq R_2'(P), \\
 R_1\leq \underset{{\tilde{P}\in\calL_0(P):  I_{\tilde{P}}(X_1;Y)\leq R_1}}{\min}   I_{\tilde{P}}(X_1;Y)+ |I_{\tilde{P}}(X_2;Y|X_1)-R_2|^+ 
 \end{array}\right\},
\end{flalign}
and
 \begin{flalign}
\left\{ (R_1,R_2):\;\begin{array}{l}
R_1\leq R_1'(P),\\
R_2\leq R_2'(P), \\
R_2 \leq \underset{\tilde{P}\in\calL_0(P):\;  I_{\tilde{P}}(X_2;Y)- I_{\tilde{P}}(X_2;X_1)\leq R_2}{\min}  I_{\tilde{P}}(X_2;Y)-I_{\tilde{P}}(X_2;X_1)+ |I_{\tilde{P}}(X_1;Y,X_2)-R_1|^+
 \end{array}\right\}.
\end{flalign}
Clearly, this union contains the region 
\begin{flalign}\label{eq: two regions union}
\left\{ (R_1,R_2):\;    \begin{array}{l}
R_1\leq R_1'(P),\\
R_2\leq R_2'(P), \\
 R_1+R_2\leq \max\left\{ \begin{array}{l}\underset{\tilde{P}\in\calL_0(P):  I_{\tilde{P}}(X_1;Y)\leq R_1}{\min}  I_{\tilde{P}}(X_1,X_1;Y),\\
 \underset{\tilde{P}\in\calL_0(P):\;  I_{\tilde{P}}(X_2;Y)- I_{\tilde{P}}(X_2;X_1)\leq R_2}{\min}  I_{\tilde{P}}(X_1,X_2;Y)\end{array}\right\}
\end{array}\right\}.
\end{flalign}
Note that $ I_{\tilde{P}}(X_1,X_1;Y)$ is convex in $\tilde{P}_{X_1,X_2|Y}$ for fixed $\tilde{P}_Y$ (and by definition of $\tilde{P}\in\calL_0(P)$, $\tilde{P}_Y=P_Y$ is fixed and not minimized) and since the sets over which the minimizations are performed are convex the  sum-rate bound in (\ref{eq: two regions union}) is equal to 
 sum-rate bound 
\begin{flalign}\label{eq: kajhfkjh.kuh}
R_1+R_2\leq\min_{\tilde{P}\in\calL_0^{bin}(P)} I_{\tilde{P}}(X_1,X_1;Y).
\end{flalign}
This yields that $\tilde{\calR}_{cog}^{bin}(P)\subseteq \calR_{cog}^{bin}(P)$, and concludes the proof of Corollary \ref{cr: Cor cor corollary}.


\end{document}